%% file: main.tex
\documentclass[a4paper,11pt]{article}
\input{Preamble.tex}

\begin{document}

\pagenumbering{Roman}

\hypersetup{pageanchor=false}
\title{Dynamic Treewidth in Logarithmic Time\thanks{The research leading to these results was funded by the European Union under Marie Skłodowska-Curie Actions (MSCA), project no. 101206430, and by VILLUM Foundation, Grant Number 16582, Basic Algorithms Research Copenhagen (BARC).}}

\author{Tuukka Korhonen\thanks{Department of Computer Science, University of Copenhagen, Denmark. \texttt{tuko@di.ku.dk}}}


\maketitle

\thispagestyle{empty}

\input{abstract.tex}

 \begin{textblock}{20}(-0.5, 7.6)
 \includegraphics[width=100px]{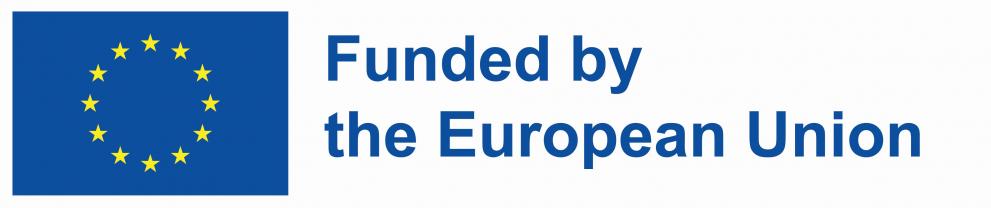}%
 \end{textblock}

\thispagestyle{empty}

\newpage

\pagenumbering{roman}

\setcounter{page}{1}
\setcounter{tocdepth}{2}

\tableofcontents

\newpage

\pagenumbering{arabic}

\hypersetup{pageanchor=true}

\clearpage
\setcounter{page}{1}

\input{intro.tex}

\input{overview.tex}

\input{prelim.tex}

\input{gt.tex}

\input{rots.tex}

\input{struct.tex}

\input{balance.tex}

\input{rotroot.tex}

\input{insdeledge.tex}

\input{together.tex}

\input{prds.tex}

\input{bwtotw.tex}

\input{conclusion.tex}

\appendix

\input{gtapp.tex}

\input{automapp.tex}

\bibliographystyle{alpha}
\bibliography{book_kernels_fvf}

\end{document}

%% file: abstract.tex
\begin{abstract}
We present a dynamic data structure that maintains a tree decomposition of width at most $9k+8$ of a dynamic graph with treewidth at most $k$, which is updated by edge insertions and deletions.
The amortized update time of our data structure is $2^{\OO(k)} \log n$, where $n$ is the number of vertices.
The data structure also supports maintaining any ``dynamic programming scheme'' on the tree decomposition, providing, for example, a dynamic version of Courcelle's theorem with $\OO_{k}(\log n)$ amortized update time; the $\OO_{k}(\cdot)$ notation hides factors that depend on $k$.
This improves upon a result of Korhonen, Majewski, Nadara, Pilipczuk, and Sokołowski~[FOCS~2023], who gave a similar data structure but with amortized update time $2^{k^{\OO(1)}} n^{o(1)}$.
Furthermore, our data structure is arguably simpler.

Our main novel idea is to maintain a tree decomposition that is ``downwards well-linked'', which allows us to implement local rotations and analysis similar to those for splay trees.
\end{abstract}

%% file: intro.tex
\section{Introduction}
\label{sec:intro}
Treewidth is one of the most well-studied graph parameters in computer science and graph theory.
Graphs of bounded treewidth generalize trees in the sense that while trees can be decomposed by separators of size $1$, graphs of treewidth $k$ can be decomposed by separators of size $k$.
Treewidth was introduced independently by multiple authors~\cite{DBLP:journals/dam/ArnborgP89,bertele1973non,Halin:1976it,DBLP:journals/jal/RobertsonS86} under various equivalent definitions.
The now-standard definition via tree decompositions was introduced by Robertson and Seymour in their graph minors series~\cite{DBLP:journals/jal/RobertsonS86}.

A \emph{tree decomposition} of a graph $G$ is a pair $(T,\bag)$, where $T$ is a tree and $\bag \colon V(T) \to 2^{V(G)}$ associates each node of $T$ with a \emph{bag} containing vertices of $G$, so that (1) for every edge $uv \in E(G)$, there is a bag containing $u$ and $v$, and (2) for every vertex $v \in V(G)$, the set of bags containing $v$ forms a non-empty connected subtree of $T$.
The width of a tree decomposition is the maximum size of a bag minus one.
The treewidth of a graph, denoted by $\tw(G)$, is the minimum width of a tree decomposition of it.
Graphs of treewidth $\le 1$ are exactly the forests, while the $n$-clique has treewidth $n-1$.

Treewidth is useful in algorithms because many graph problems that are hard in general become efficiently solvable on graphs of bounded treewidth via dynamic programming on a tree decomposition.
For example, there are algorithms with running time $2^{\OO(k)} n$, where $k$ is the width of a tree decomposition and $n$ is the number of vertices, for problems such as $3$-coloring, maximum independent set, minimum dominating set, and Hamiltonicity~\cite{DBLP:journals/dam/ArnborgP89,DBLP:conf/icalp/Bodlaender88,BodlaenderCK12,TelleP97}.
Moreover, the celebrated Courcelle's theorem~\cite{DBLP:journals/iandc/Courcelle90} (see also~\cite{ArnborgLS91,BoriePT92,Courcelle:2012book}) gives $\OO_k(n)$ time\footnote{The $\OO_k(\cdot)$-notation hides factors depending on $k$.} algorithms for all graph problems expressible in the counting monadic second-order logic~($\CMSO_2$).
Furthermore, treewidth is frequently used as a tool for solving problems even on graphs of large treewidth~\cite{Baker94,DemaineFHT05jacm,RobertsonS-GMXIII}, and the applications of treewidth are not limited only to graph problems~\cite{DBLP:journals/tcs/ChekuriR00,DBLP:conf/stoc/DongLY21,lauritzen1988local,DBLP:journals/siamcomp/MarkovS08}.

\paragraph{Computing treewidth.}
The algorithmic applications of treewidth require not only the input graph to have small treewidth, but also a tree decomposition of small width to be given.
This makes the problem of computing a tree decomposition of small width, if one exists, the central algorithmic problem in this context.
Let us mention a few of the over 20 publications on this problem.
While computing treewidth is NP-complete in general~\cite{ArnborgCP87}, there are algorithms for computing a tree decomposition of optimal width with running times $\OO(n^{k+2})$~\cite{ArnborgCP87}, $2^{\OO(k^3)} n$~\cite{DBLP:journals/siamcomp/Bodlaender96}, and $2^{\OO(k^2)} n^4$~\cite{DBLP:conf/stoc/KorhonenL23}.
No $2^{o(n)}$ time algorithms exist assuming the Exponential Time Hypothesis~(ETH)~\cite{DBLP:journals/corr/abs-2406-11628}.

As for constant-approximation, the classic Robertson-Seymour algorithm finds a $4$-ap\-prox\-i\-mate\-ly optimal tree decomposition in time $\OO(3^{3k} \cdot k^2 \cdot n^2)$~\cite{RobertsonS-GMXIII}, and the recent algorithm of Korhonen a $2$-approximation in time $2^{\OO(k)} n$~\cite{Korhonen21} (see also~\cite{BodlanderDDFLP13}).
Assuming the Small Set Expansion (SSE) hypothesis~\cite{DBLP:conf/stoc/RaghavendraS10}, no polynomial-time constant-approximation algorithms exist~\cite{WuAPL14}.
The best known approximation ratio in polynomial-time is $\OO(\sqrt{\log k})$~\cite{FeigeHL08}, and in time $k^{\OO(1)} n \polylog n$ one can achieve approximation ratios of $\OO(k)$~\cite{DBLP:journals/talg/FominLSPW18} and $\OO(\log n)$~\cite{DBLP:conf/esa/DongY24}.

\paragraph{Dynamic treewidth.}
In this work we consider the problem of computing treewidth in the dynamic setting.
In particular, the goal is to have a data structure that maintains a dynamic graph $G$ under edge insertions and deletions, and an approximately optimal tree decomposition of $G$.
Furthermore, we would like to simultaneously maintain any dynamic programming scheme on the tree decomposition, to lift the numerous applications of treewidth to the dynamic setting.
This question of dynamic data structures for treewidth can be regarded as the common generalization of two active research areas: algorithms for computing treewidth, and data structures for dynamic forests.
In particular, dynamic forests are the case of treewidth~$1$.

\paragraph{Dynamic treewidth $1$.}
Sleator and Tarjan~\cite{SleatorT83} gave a data structure for dynamic forests with $\OO(\log n)$ worst-case update time, called the \emph{link-cut tree}.
Among other applications, they used link-cut trees to obtain a faster algorithm for maximum flow.
Frederickson~\cite{DBLP:journals/siamcomp/Frederickson85,DBLP:journals/siamcomp/Frederickson97} introduced \emph{topology} trees, which have $\OO(\log n)$ worst-case update time, and applied them, for example, to dynamic minimum spanning trees.
Alstrup, Holm, de Lichtenberg, and Thorup~\cite{DBLP:journals/talg/AlstrupHLT05} introduced \emph{top trees}, which also have $\OO(\log n)$ worst-case update time, and used them for improved algorithms for problems such as dynamic connectivity~\cite{DBLP:journals/jacm/HolmLT01}.
One can interpret the top tree as maintaining a tree decomposition of width at most $2$ and depth at most $\OO(\log n)$, while providing an interface for maintaining any dynamic programming scheme on the tree decomposition.\footnote{In fact, top trees correspond to \emph{branch decompositions}~\cite{RobertsonS91} of width $2$, which can be interpreted as tree decompositions of width $2$.}
P{\u{a}}tra{\c{s}}cu and Demaine~\cite{DBLP:journals/siamcomp/PatrascuD06} showed that dynamic connectivity on forests requires $\Omega(\log n)$ update time, even when amortization is allowed, establishing the optimality of the aforementioned data structures.
Other works on dynamic forests include~\cite{DBLP:journals/jal/Frederickson97,DBLP:conf/sosa/HolmRR23,DBLP:conf/soda/TarjanW05}.

\paragraph{Dynamic treewidth $2$ and $3$.}
The first work on the dynamic treewidth problem was by Bodlaender~\cite{DBLP:conf/wg/Bodlaender93a}, who gave a dynamic data structure to maintain a tree decomposition of width at most~$11$ of a dynamic graph of treewidth at most $2$.
His data structure has $\OO(\log n)$ worst-case update time and supports maintaining arbitrary dynamic programming schemes.
It is based on Frederickson's approach~\cite{DBLP:journals/jal/Frederickson97} for dynamic trees.
Bodlaender also observed that in the decremental setting, i.e., without edge insertions, achieving $\OO_k(\log n)$ update time for treewidth $k$ is rather trivial.
Concurrently with Bodlaender, Cohen, Sairam, Tamassia, and Vitter~\cite{CohenSTV93}\footnote{We were not able to access~\cite{CohenSTV93}, so therefore our description of it is based on that of Bodlaender~\cite{DBLP:conf/wg/Bodlaender93a}.} gave an $\OO(\log^2 n)$ update time algorithm for the treewidth $2$ case and an $\OO(\log n)$ update time algorithm for the treewidth $3$ case in the incremental setting, i.e., without edge deletions.

\paragraph{Dynamic treewidth $k$.}
Bodlaender~\cite{DBLP:conf/wg/Bodlaender93a} asked whether an $\OO_k(\log n)$ update time dynamic treewidth data structure could be obtained for graphs of treewidth at most $k$.
There was little to no progress on this question for almost 30 years.
During this time, authors considered dynamic data structures for graph parameters larger than treewidth, such as treedepth and feedback vertex number~\cite{AlmanMW20,ChenCDFHNPPSWZ21,DvorakKT14,DBLP:conf/stacs/MajewskiPS23}, and models of dynamic treewidth where the tree decomposition does not change or the changes are directly specified as input~\cite{Frederickson98,DBLP:journals/algorithmica/Hagerup00}.
The question of dynamic treewidth was repeatedly re-stated~\cite{AlmanMW20,ChenCDFHNPPSWZ21,DBLP:conf/stacs/MajewskiPS23}.

The first dynamic treewidth data structure that works for any treewidth bound $k$ was obtained by Goranci, R{\"{a}}cke, Saranurak, and Tan~\cite{GoranciRST21} as an application of their dynamic expander hierarchy data structure.
It has a subpolynomial $n^{o(1)}$ update time, and maintains an $n^{o(1)}$-factor approximately optimal tree decomposition, but works only for graphs with maximum degree $n^{o(1)}$.
As the width of the decomposition maintained can be superlogarithmic in $n$ even for graphs of constant treewidth, this data structure is not suitable for most applications of treewidth, which use dynamic programming with running time exponential in the width.

Recently, Korhonen, Majewski, Nadara, Pilipczuk, and Sokołowski~\cite{DBLP:conf/focs/KorhonenMNP023} gave the first dynamic treewidth data structure that maintains a tree decomposition whose width is bounded by a function of only $k$, and that has amortized update time sublinear in $n$ for every fixed $k$.
In particular, their data structure maintains a tree decomposition of width at most $6k+5$, with amortized update time $2^{k^{\OO(1)} \sqrt{\log n \log \log n}} = 2^{k^{\OO(1)}} n^{o(1)}$.
Furthermore, dynamic programming schemes can be maintained on the tree decomposition with a similar running time, with an overhead depending only on the per-node running time of the scheme, for example, $2^{\OO(k)}$ for $3$-coloring and maximum independent set.

The data structure of~\cite{DBLP:conf/focs/KorhonenMNP023} has been already applied by Korhonen, Pilipczuk, and Stamoulis~\cite{DBLP:conf/focs/KorhonenPS24} to obtain an almost-linear $\OO_H(n^{1+o(1)})$ time algorithm for $H$-minor containment, improving upon an $\OO_H(n^2)$ time algorithm of~\cite{DBLP:journals/jct/KawarabayashiKR12}.
It was also generalized by Korhonen and Sokołowski~\cite{DBLP:conf/stoc/Korhonen024} to the setting of rankwidth, yielding also an improved algorithm for computing rankwidth in the static setting.

\paragraph{Our contribution.}
In this work, we resolve the question of Bodlaender~\cite{DBLP:conf/wg/Bodlaender93a} by giving a dynamic treewidth data structure with arguably optimal amortized update time.
Our data structure maintains a tree decomposition of width at most $9 \cdot \tw(G) + 8$ and has amortized update time $2^{\OO(k)} \log n$, where $k$ is an upper bound for the treewidth of the dynamic graph $G$, given at the initialization.
Furthermore, it supports the maintenance of arbitrary dynamic programming schemes, similarly to the data structure of~\cite{DBLP:conf/focs/KorhonenMNP023}.
The update time $2^{\OO(k)} \log n$ is arguably optimal in the sense that dynamic forests require $\Omega(\log n)$ time~\cite{DBLP:journals/siamcomp/PatrascuD06}, and all known constant-approximation algorithms for treewidth have a factor of $2^{\OO(k)}$ in their running time.

To more formally state our main result, let us introduce some notation.
A \emph{tree decomposition automaton} is, informally speaking, an automaton that implements bottom-up dynamic programming on rooted tree decompositions.
For example, there exists a tree decomposition automaton for deciding whether a graph is $3$-colorable, whose \emph{evaluation time}, i.e., time spent per node, is $\evaltime(k) = 2^{\OO(k)}$, where $k$ is the width of the tree decomposition.
A \emph{rooted tree decomposition} is a tree decomposition $(T,\bag)$ where $T$ is a rooted tree.
The \emph{depth} of $(T,\bag)$ is the maximum length of a root-leaf path, and $(T,\bag)$ is \emph{binary} if each node of $T$ has at most two children.

\begin{restatable}{theorem}{maintheorem}
\label{the:maintheorem}
There is a data structure that is initialized with an edgeless $n$-vertex graph $G$ and an integer $k$, supports updating $G$ via edge insertions and deletions under the promise that $\tw(G) \le k$ at all times, and maintains a rooted tree decomposition of $G$ of width at most $9 \cdot \tw(G) + 8$.
The amortized running time of the initialization is $2^{\OO(k)} n$, and the amortized running time of each update is $2^{\OO(k)} \log n$.

Moreover, if at the initialization the data structure is provided a tree decomposition automaton $\autom$ with evaluation time $\evaltime$, then a run of $\autom$ on the tree decomposition is maintained, incurring an additional $\evaltime(9k+8)$ factor on the running times.

Furthermore, the tree decomposition is binary and its depth is bounded by $2^{\OO(k)} \log n$.
\end{restatable}

We note that the statement of \Cref{the:maintheorem} could be strengthened in various ways, but we prefer to not overload this paper with technical extensions of it.
We will discuss the possible strengthenings of \Cref{the:maintheorem} in the Conclusions section (\Cref{sec:conclusion}).

\paragraph{Applications.}
By combining \Cref{the:maintheorem} with well-known dynamic programming procedures for graphs of bounded treewidth~\cite{DBLP:journals/dam/ArnborgP89,ArnborgLS91,DBLP:conf/icalp/Bodlaender88,BoriePT92,DBLP:journals/iandc/Courcelle90,TelleP97}, we obtain the following corollary.

\begin{restatable}{corollary}{cordp}
\label{cor:twdp}
On fully dynamic $n$-vertex graphs of treewidth at most $k$, there are
\begin{itemize}
\item $2^{\OO(k)} \log n$ amortized update time dynamic algorithms for maintaining the size of a maximum independent set, the size of a minimum dominating set, $q$-colorability for constant $q$, etc., and
\item $\OO_k(\log n)$ amortized update time dynamic algorithms for maintaining any graph property expressible in the counting monadic second-order logic.
\end{itemize}
\end{restatable}

Furthermore, while the data structure of \Cref{the:maintheorem} requires a pre-set treewidth bound $k$, it does at all points maintain a $9$-approximately optimal tree decomposition of the current graph.
By using a tree decomposition automaton for exact computing of treewidth based on the work of Bodlaender and Kloks~\cite{BodlaenderK96}, as was done in~\cite{DBLP:conf/focs/KorhonenMNP023}, it could also maintain the exact value of treewidth with the cost of a $2^{\OO(k^3)}$ factor in the update time.

By exploiting the known graph-theoretical properties of treewidth, \Cref{the:maintheorem} yields several direct consequences to the growing area of dynamic parameterized algorithms.
For example, treewidth can be applied via the grid minor theorem~\cite{ChekuriC13,DBLP:journals/jct/RobertsonS86} and its versions for planar and minor-free graphs~\cite{DemaineH08,RobSeymT94}.
Let us first state a couple of consequences that improve upon existing problem-specific dynamic parameterized algorithms.

\begin{restatable}{corollary}{corkpathfvs}
\label{cor:corkpathfvs}
For fully dynamic $n$-vertex graphs, there is a dynamic data structure that is given a parameter $k$ at the initialization, has $2^{\OO(k)} \log n$ amortized update time, and maintains
\begin{itemize}
\item whether the graph contains a path of length $k$, and
\item whether the graph contains a feedback vertex set of size at most $k$.
\end{itemize}
\end{restatable}

The result on $k$-path can be compared to several earlier dynamic parameterized algorithms: Dvo\v{r}{\'{a}}k et al.~\cite{DvorakKT14} gave a dynamic data structure with $f(k)$ amortized update time for non-elementary $f$, Alman et al.~\cite{AlmanMW20} with $2^{\OO(k \log k)} \textsf{DC}(n)$ update time, where $\textsc{DC}(n)$ is the update time for dynamic connectivity on general graphs\footnote{Known bounds for $\textsc{DC}(n)$ include expected amortized time $\OO(\log n (\log \log n)^2)$~\cite{DBLP:journals/theoretics/HuangHKPT23}, deterministic amortized time $\OO(\log^2 n/\log \log n)$~\cite{DBLP:conf/soda/Wulff-Nilsen13a}, and deterministic worst-case time $n^{o(1)}$~\cite{DBLP:conf/focs/ChuzhoyGLNPS20}.}, and Chen et al.~\cite{ChenCDFHNPPSWZ21} with $2^{\OO(k^2)}$ amortized update time.
The result on feedback vertex set improves over the result of Alman et al.~\cite{AlmanMW20} that has $2^{\OO(k \log k)} \log n$ amortized update time.

As further example applications, we observe that we obtain dynamic subexponential parameterized algorithms on planar graphs via the framework of bidimensionality~\cite{DemaineFHT05jacm}.

\begin{restatable}{corollary}{corsubexp}
\label{cor:corsubexp}
For fully dynamic planar $n$-vertex graphs, there is a dynamic data structure that is given a parameter $k$ at the initialization, has $2^{\OO(\sqrt{k})} \log n$ amortized update time, and maintains
\begin{itemize}
\item whether the graph has a dominating set of size at most $k$, and
\item whether the graph contains a path of length at least $k$.
\end{itemize}
\end{restatable}

We defer further discussions about the applications of dynamic treewidth and future directions to the Conclusions section (\Cref{sec:conclusion}).
We suggest an interested reader to also take a look at the Introduction and Conclusions sections of~\cite{DBLP:conf/focs/KorhonenMNP023} for more applications of dynamic treewidth.

\paragraph{Our techniques.}
Our main novel insight is to maintain a rooted tree decomposition $\Tc = (T,\bag)$ that is ``downwards well-linked''.
This means that for any node $t$ of $\Tc$ and its parent $p$, if we consider the \emph{adhesion} $\adh(tp) = \bag(t) \cap \bag(p)$ of the edge $tp$, and take two subsets $A,B \subseteq \adh(tp)$ of the same size, we can route $|A|=|B|$ vertex-disjoint paths from $A$ to $B$, using only vertices in the subtree of $\Tc$ rooted at $t$.
Technically, this will be formulated through what we call ``downwards well-linked superbranch decompositions'' instead of tree decompositions.
However, we use tree decompositions here for simplicity.

This condition of downwards well-linkedness is useful in that it directly guarantees that the size of every adhesion $\adh(tp)$ is at most $\OO(\tw(G))$.
Furthermore, it allows us to lift local properties in the bags of the tree decomposition to global properties in the graph.
In particular, it allows us to conclude that whenever a bag is too large, particularly, larger than some bound in $2^{\OO(k)}$, we can locally split it into two bags, while maintaining downwards well-linkedness.
Moreover, this splitting cleanly partitions the children of the bag as the children of the two resulting nodes, and allows us to control which children are pushed downwards in the tree and which stay at the current depth.

In addition to this splitting operation, we can also contract two adjacent nodes of the tree decomposition into one.
With these splitting and contraction operations, we arrive at an abstract dynamic tree maintenance problem in which we manipulate a tree by either splitting nodes of high degree or contracting edges.
Our goal is to maintain a tree of depth $2^{\OO(k)} \log n$ and maximum degree $2^{\OO(k)}$.
It turns out that our operations suffice to implement manipulations similar to those of splay trees~\cite{DBLP:journals/jacm/SleatorT85}, and we indeed manage to solve this tree maintenance problem using an analysis similar to that for splay trees.


Our approach is completely different compared to the approach of~\cite{DBLP:conf/focs/KorhonenMNP023}, but we use the framework of ``prefix-rebuilding updates'' introduced in~\cite{DBLP:conf/focs/KorhonenMNP023} for formalizing updates to dynamic tree decompositions.
The key concept of downwards well-linkedness is directly from the recent work of Korhonen~\cite{DBLP:journals/corr/abs-2411-02658}.
In hindsight, it can be regarded as a generalization of invariants used for topology trees and top trees~\cite{DBLP:journals/talg/AlstrupHLT05,DBLP:journals/siamcomp/Frederickson85}.
Similar concepts have also been used in the context of mimicking networks~\cite{DBLP:conf/soda/ChalermsookDKLL21} and expander decompositions~\cite{GoranciRST21}.

\paragraph{Organization.}
In \Cref{sec:overview} we present an informal sketch of the proof of \Cref{the:maintheorem}.
The proof is presented in detail through \Cref{sec:preli,sec:downwl,sec:mani,sec:struct,sec:balancing,sec:insdeledge,sec:mainlemmaproof,sec:prds}.
In \Cref{sec:preli} we present definitions and preliminary results.
Then, in \Cref{sec:downwl,sec:mani,sec:struct,sec:balancing,sec:insdeledge,sec:mainlemmaproof} we present the core part of our data structure, which is about maintaining a so-called ``downwards well-linked superbranch decomposition''.
The main graph-theoretical properties of these decompositions are discussed in \Cref{sec:downwl} and subroutines for manipulating them in \Cref{sec:mani}.
In \Cref{sec:struct} we introduce the invariants of our dynamic data structure and state the main lemma about it, which is then proven in \Cref{sec:balancing,sec:insdeledge,sec:mainlemmaproof}.
In \Cref{sec:prds} we lift our data structure from the setting of downwards well-linked superbranch decompositions to that of treewidth.
We discuss conclusions and future directions in \Cref{sec:conclusion}.

\paragraph{AI declaration.}
AI was used only for proofreading the manuscript.

%% file: overview.tex
\section{Outline}
\label{sec:overview}
We present a sketch of the proof of \Cref{the:maintheorem}.
The proof will be presented in full detail in \Cref{sec:preli,sec:downwl,sec:mani,sec:struct,sec:balancing,sec:insdeledge,sec:mainlemmaproof,sec:prds}.

\paragraph{Downwards well-linked superbranch decompositions.}
The core idea of this work is to maintain a so-called ``downwards well-linked superbranch decomposition'' of the dynamic graph $G$, so let us start by defining it.
A \emph{superbranch decomposition} of a graph $G$ is a pair $\Tc = (T,\lmap)$, where $T$ is a rooted tree in which every non-leaf node has at least two children and $\lmap$ is a bijection from the leaves of $T$ to $E(G)$.\footnote{The term ``superbranch decomposition'' was introduced by~\cite{DBLP:journals/corr/abs-2411-02658}. Superbranch decompositions are like branch decompositions of Robertson and Seymour~\cite{RobertsonS91}, but allow nodes with more than two children.}
For a node $t \in V(T)$, let us denote by $\lmap[t] \subseteq E(G)$ the edges of $G$ that are mapped to the leaves of $T$ that are in the subtree rooted at $t$.

The \emph{boundary} of a set $A \subseteq E(G)$ of edges of $G$ is the set $\bd(A) \subseteq V(G)$ consisting of the vertices that are incident to edges in both $A$ and $E(G) \setminus A$.
We denote the boundary size by $\bdc(A) = |\bd(A)|$.
We say that a set $A \subseteq E(G)$ of edges of $G$ is \emph{well-linked} if there is no bipartition $(C_1,C_2)$ of $A$ so that $\bdc(C_i) < \bdc(A)$ for both $i \in \{1,2\}$.
Equivalently, $A$ is well-linked if for any two subsets $B_1,B_2 \subseteq \bd(A)$ of the same size, possibly overlapping, we can connect $B_1$ to $B_2$ by $|B_1|=|B_2|$ vertex-disjoint paths that use only edges in $A$.
Now, a superbranch decomposition is \emph{downwards well-linked} if for every node $t$, the set $\lmap[t]$ is well-linked.

Our main goal is to maintain a superbranch decomposition of the dynamic graph $G$, that is downwards well-linked and has maximum degree $2^{\OO(k)}$, where $k$ is the upper bound on the treewidth of $G$.
But how are superbranch decompositions and downwards well-linkedness related to treewidth?
For each edge $tp \in E(T)$ of the superbranch decomposition, where $p$ is the parent of $t$, we define that the \emph{adhesion} at $tp$ is the set $\adh(tp) = \bd(\lmap[t])$.
It follows from well-known connections between treewidth and well-linkedness~\cite{RobertsonS-GMXIII} that if $A \subseteq E(G)$ is a well-linked set, then $\bdc(A) \le 3 \cdot \tw(G) + 3$.
Therefore, if $\Tc$ is downwards well-linked, then each of its adhesions has size at most $3 \cdot \tw(G) + 3$.
We can view $\Tc$ as a tree decomposition of $G$ by associating each node $t \in V(T)$ with a bag $\bag(t)$ consisting of the union of the adhesions at edges incident to $t$.
It turns out that the resulting pair $(T,\bag)$ is indeed a tree decomposition of $G$.\footnote{This is not formally true, such a tree decomposition can fail to include vertices of degree less than $2$. We handle this with specific definitions in the technical part of the paper, but it could also be handled by assuming that the minimum degree is at least $2$, e.g., by adding pendant triangles to every vertex.}

It follows that if $\Tc$ has degree at most $\Delta$, then it corresponds to a tree decomposition of width at most $\OO(\Delta \cdot \tw(G))$.
Therefore, by maintaining a downwards well-linked superbranch decomposition of degree at most $2^{\OO(k)}$, we manage to maintain a tree decomposition of width at most $2^{\OO(k)}$.
This falls short of the goal of maintaining a tree decomposition of width at most $9 \tw(G) + 8$, but we can apply a local ``post-processing'' step on top of a downwards well-linked superbranch decomposition to convert it into a tree decomposition of width at most $9k+8$.
Let us return to this post-processing step at the end of this proof sketch, and for now just focus on the goal of maintaining a downwards well-linked superbranch decomposition with degree at most $2^{\OO(k)}$, with amortized update time $2^{\OO(k)} \log n$.

\paragraph{Operations on downwards well-linked superbranch decompositions.}
The main utility of downwards well-linkedness is that it allows us to translate local properties in the nodes of the superbranch decomposition to global properties in the graph, enabling us to implement local rotations.
Let us introduce some notation to state more clearly what we mean by this.

A \emph{hypergraph} is a graph that has \emph{hyperedges} instead of edges, where hyperedges correspond to arbitrary subsets of vertices instead of pairs of vertices.
For a hypergraph $G$ and a hyperedge $e \in E(G)$, we denote by $V(e) \subseteq V(G)$ the set of vertices of $e$.
We allow a hypergraph to contain multiple hyperedges corresponding to the same set of vertices, in particular, there can be $e_1 \neq e_2$ with $V(e_1) = V(e_2)$.
We define the boundary $\bd(A)$ of a set $A \subseteq E(G)$ of hyperedges in the similar way as we defined it for edges, i.e., as $\bd(A) = \left(\cup_{e \in A} V(e)\right) \cap \left(\cup_{e \in E(G) \setminus A} V(e)\right)$.
We denote $\bdc(A) = |\bd(A)|$ also in this context.
We also define well-linkedness in the same way, i.e., $A$ is well-linked if there is no bipartition $(C_1,C_2)$ of $A$ so that $\bdc(C_i) < \bdc(A)$ for both $i \in \{1,2\}$.

One more definition we need is that of a \emph{torso} of a node of a superbranch decomposition.
The torso $\torso(t)$ of a node $t \in V(T)$ is the hypergraph that has a hyperedge $e_s$ for each neighbor $s$ of $t$ in the tree $T$.
The vertex set of the hyperedge $e_s$ is the adhesion at the edge $st$, i.e., $V(e_s) = \adh(st)$.
The vertex set of $\torso(t)$ is the union of the vertex sets of its hyperedges, i.e., the union of the adhesions at the edges incident to $t$.

Now, let $A \subseteq E(\torso(t))$ be a set of hyperedges in $\torso(t)$, and assume that $A$ does not contain the hyperedge $e_p$ corresponding to the edge between $t$ and its parent $p$.
Now, $A$ corresponds to a set of children $\mathcal{C}_A = \{c \mid e_c \in A\}$ of $t$, which in turn corresponds to a set of edges $\bigcup_{c \in \mathcal{C}_A} \lmap[c]$ of $G$.
We denote this set of edges of $G$ corresponding to $A$ by $A \orescliqs \Tc$.
The following is the key lemma that enables us to lift well-linkedness in torsos to well-linkedness in $G$.

\begin{lemma}[Informal version of \Cref{lem:wltransindecomp}]
\label{lem:ovwllem}
If $\Tc$ is downwards well-linked, and $A \subseteq E(\torso(t))$ does not contain the hyperedge $e_p$ corresponding to the hyperedge between $t$ and its parent $p$, then $A$ is well-linked in $\torso(t)$ if and only if $A \orescliqs \Tc$ is well-linked in $G$.
\end{lemma}

The proof of \Cref{lem:ovwllem} follows from a lemma proven in~\cite{DBLP:journals/corr/abs-2411-02658}, but it is not very hard to prove from scratch by using the submodularity of the boundary size function $\bdc$.

\begin{figure}[!b]
\begin{center}
\includegraphics[width=0.75\textwidth]{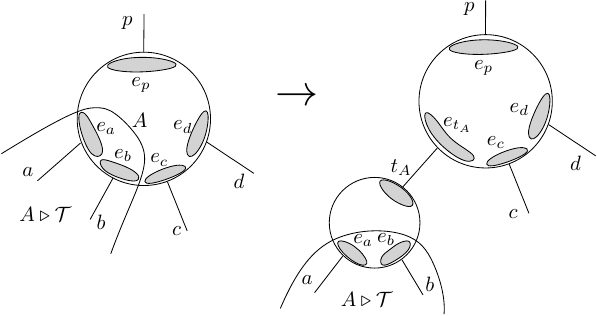}
\caption{Splitting a node $t$ of a superbranch decomposition $\Tc$ using a well-linked set $A = \{e_a, e_b\}$ of $\torso(t)$, which corresponds to a well-linked set $A \orescliqs \Tc$ of $G$.\label{fig:splitting}}
\end{center}
\end{figure}

\Cref{lem:ovwllem} enables us to implement the operation of splitting nodes of $\Tc$.
In particular, suppose that $A \subseteq E(\torso(t))$ does not contain $e_p$, is well-linked in $\torso(t)$, has size at least $|A| \ge 2$, and does not contain all children of $t$.
Then, we can use $A$ to split $t$ into two nodes, $t_A$ and $t'$, so that the children of $t_A$ will be children $\mathcal{C}_A = \{c \mid e_c \in A\}$ of $t$, and the children of $t'$ will be $t_A$ and the other children of $t$.
This results in a downwards well-linked superbranch decomposition because $A \orescliqs \Tc$ is downwards well-linked by \Cref{lem:ovwllem}.
See \Cref{fig:splitting} for an illustration.

Now, all we need to do to reduce the degree of a node $t$ is to find a well-linked set $A$ in $\torso(t)$, so that $e_p \notin A$ and $2 \le |A| < \ddeg(t)$, where $\ddeg(t)$ denotes the number of children of $t$.
For this, we use the following lemma.

\begin{lemma}[Corresponds to \Cref{lem:partitiontowlalg}]
\label{lem:ovwlfind}
Given a hypergraph $G$ and a set $X \subseteq E(G)$, we can in time $2^{\OO(\bdc(X))} \cdot \|G\|^{\OO(1)}$ find a partition $\compset$ of $X$ into at most $2^{\bdc(X)}$ well-linked sets.
\end{lemma}

The idea of the proof of \Cref{lem:ovwlfind} is that if $X$ is not already well-linked, then by definition we can partition $X$ into $(C_1,C_2)$ so that $\bdc(C_i) < \bdc(X)$.
We iteratively continue partitioning the parts $C_i$ until they are all well-linked, noting that the measure $\sum_i 2^{\bdc(C_i)}$ does not increase in this process.
Therefore, we end up with a partition into at most $2^{\bdc(X)}$ well-linked sets.

Now, to find a desired well-linked set $A$, assuming $t$ has high enough degree, it suffices to simply take $X = E(\torso(t)) \setminus \{e_p, e_i\}$, where $e_i$ is an arbitrary hyperedge of $\torso(t)$ other than $e_p$, and apply \Cref{lem:ovwlfind} with $X$.
This is guaranteed to find at least one part of size $\ge 2$ if $|X| > 2^{\bdc(X)}$.
We can bound 
\begin{equation}
\label{equ:ovsplitting}
\bdc(X) = \bdc(\{e_p, e_i\}) \le |V(e_p)| + |V(e_i)| \le 2 \cdot \adhsize(\Tc) \le 6 \cdot \tw(G) + 6,
\end{equation}
so this is successful whenever $|X| > 2^{6 \cdot \tw(G) + 6}$, i.e., whenever $\ddeg(t) \ge 2 + 2^{6 \cdot \tw(G) + 6}$.

This splitting strategy forms the core of how the maximum degree $2^{\OO(k)}$ is maintained.
Additionally, it gives some control on how the superbranch decomposition changes.
In particular, by the choice of the hyperedge $e_i$, we can pick a child of $t$ that is guaranteed to not be pushed deeper down in the tree by the splitting operation.
The argument of \Cref{equ:ovsplitting} in fact generalizes to picking multiple children, with the cost of a higher constant factor.
In our algorithm we will use it with at most $3$ children.

In addition to the splitting operation, the other basic operation we use for manipulating downwards well-linked superbranch decompositions is the contraction operation.
This simply means contracting an edge (that is not adjacent to a leaf) of the superbranch decomposition.
It is straightforward to see that contraction always preserves downwards well-linkedness.

\paragraph{Balancing.}
Before explaining how we implement the operations of adding and deleting edges, let us focus on how we keep the superbranch decomposition balanced.
We will maintain that the superbranch decomposition always has depth at most $2^{\OO(k)} \log n$, and analyze the work used for balancing the decomposition by using a potential function similar to the potential function of splay trees~\cite{DBLP:journals/jacm/SleatorT85}.
We could have also taken the splay tree approach of allowing an unbalanced tree and analyzing all the work via potential, but to us the approach of maintaining a depth upper bound felt more natural.
Furthermore, in some applications of treewidth (e.g.~\cite{DBLP:conf/icalp/Lampis14}) logarithmic-depth decompositions are required, so it could be useful that our data structure directly provides them.

For a parameter $d$, we call a node $t$ \emph{$d$-unbalanced} if it has a descendant $s$ at distance $d$ so that $|\lmap[s]| \ge \frac{2}{3} |\lmap[t]|$.
We will maintain that for some $d = 2^{\OO(k)}$, our superbranch decomposition contains no $d$-unbalanced nodes.
It is easy to see that this implies that the depth is at most $\OO(d \log n) = 2^{\OO(k)} \log n$.

The main idea is to introduce a potential function, so that whenever the decomposition contains a $d$-unbalanced node, we can improve it by applying splitting and contraction operations, decreasing the potential, and pay for the work done through this decrease.
The potential function we use is
\[\pot(\Tc) = \sum_{t \in \vint(T)} (\ddeg(t)-1) \cdot \log(|\lmap[t]|).\]
Here, $\vint(T)$ denotes the internal (i.e., non-leaf) nodes of $T$, $\ddeg(t)$ the number of children of $t$, and $\lmap[t]$ the set of leaf-descendants of $t$.

\begin{figure}[!t]
\begin{center}
\includegraphics[width=0.8\textwidth]{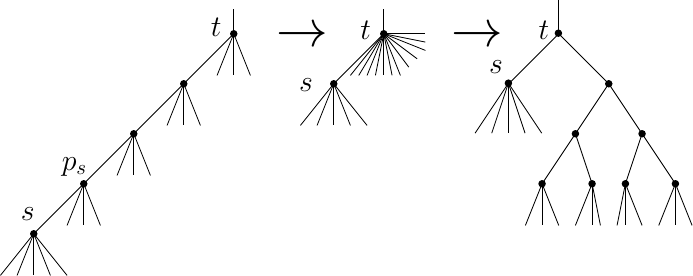}
\caption{The balancing subroutine.\label{fig:ovbalancing}}
\end{center}
\end{figure}

Now, if $t$ is a $d$-unbalanced node with a descendant $s$ at distance $d$ so that $|\lmap[s]| \ge \frac{2}{3} |\lmap[t]|$, we edit $\Tc$ as follows.
Let $p_s$ be the parent of $s$.
We contract the path from $t$ to $p_s$ into one node.
Now, $s$ becomes a child of $t$, and $t$ will have at least $d$ children.
Here, we choose $d$ large enough depending on the bound for the number of children in the splitting operation.
We then apply the splitting operation iteratively to again decrease the degree of $t$ (and also the degree of the just created descendants of $t$), but so that $s$ always remains a child of $t$ and is not pushed further down in the tree.
This process re-builds the subtree that consisted of the long $t$-$s$-path into a subtree in which $s$ is a child of $t$.
This keeps the sum of the terms $(\ddeg(t)-1)$ in the subtree unchanged, but decreases the quantity $\log(|\lmap[t]|)$ by a constant for a significant factor of the nodes.
See \Cref{fig:ovbalancing}.
By choosing $d = 2^{\OO(k)}$ large enough, the potential $\pot(\Tc)$ decreases by a constant, and the operation can be implemented in time $2^{\OO(k)}$.

The overall blueprint of how our data structure works is that in each edge insertion and deletion, the superbranch decomposition is edited for a part consisting of only $2^{\OO(k)} \log n$ nodes, so that the potential increases by at most $2^{\OO(k)} \log n$.
After these edits, which do not necessarily maintain balance, we apply the balancing subroutine as long as there are $d$-unbalanced nodes.
For this, we need to also efficiently find $d$-unbalanced nodes in a changing superbranch decomposition, but this is not hard to do by maintaining a queue storing nodes that have changed, and for each node $t$ the quantity $|\lmap[t]|$.

\paragraph{Inserting and deleting edges.}
Finally, let us explain how our data structure supports updating the graph $G$ by edge insertions and deletions.
For this, we have to admit that we do not actually maintain a superbranch decomposition of $G$, but a superbranch decomposition of the \emph{support hypergraph} of $G$, denoted by $\su{G}$.
The hypergraph $\su{G}$ has a hyperedge $e_v$ for each vertex $v \in V(G)$, with $V(e_v) = \{v\}$.
Naturally, it also has a hyperedge $e_{uv}$ with $V(e_{uv}) = \{u,v\}$ for each edge $uv \in E(G)$.\footnote{In our actual definition of the support hypergraph, it also has a hyperedge $e_{\bot}$ with $V(e_{\bot}) = \emptyset$, but this is only for resolving technicalities that do not show up in this proof sketch, so we do not include $e_{\bot}$ here.}
A superbranch decomposition of a hypergraph is defined in the exactly same way as a superbranch decomposition of a graph, but having hyperedges instead of edges mapped to the leaves of the tree.
Using the support hypergraph of $G$ instead of $G$ itself has the advantage that for each vertex $v$, there is a fixed hyperedge $e_v$ containing $v$.
Furthermore, it makes no difference to the relations between treewidth and well-linkedness.

The plan for edge insertions and deletions is the following:
For inserting an edge $uv$, we first rotate the leaves $\lmap^{-1}(e_u)$ and $\lmap^{-1}(e_v)$ corresponding to $e_u$ and $e_v$ up in the tree so that they become children of the root.
We then add a leaf corresponding to the new edge $uv$ as another child of the root.
For deleting an edge $uv$, we rotate each of the leaves $\lmap^{-1}(e_u)$, $\lmap^{-1}(e_v)$, and $\lmap^{-1}(e_{uv})$ up to become children of the root, and then delete the leaf $\lmap^{-1}(e_{uv})$.
The advantage of this process is that once these leaves are rotated to become children of the root, the edge insertion and deletion operations become trivial: Nothing needs to be recomputed in the decomposition outside of the root and its children, and downwards well-linkedness is maintained.

It remains to design a subroutine to rotate a set of at most $3$ given leaves so that they become children of the root node.
To decrease the depth of a leaf $\ell$, we only need to contract the edge between the parent $p$ and the grandparent $g$ of $\ell$.
This may result in a node of too large degree, but if it does, we simply apply the splitting operation, and in such a way that it does not again increase the depth of the leaf $\ell$.
In this way, we manage to decrease the depth of $\ell$ while maintaining the invariants, except that we may create unbalanced nodes which we will handle afterwards.
A single step like this can be implemented with one contraction operation and at most $2^{\OO(k)}$ split operations, running in time $2^{\OO(k)}$.
We also need to take into account that we may be rotating up to $3$ leaves towards the root simultaneously, and thus must avoid pushing other leaves downward while rotating one upward.
This can be resolved by always rotating the deepest leaf up, and in the splitting phase specifying these leaves as the special children, if they are children of the node we are splitting.

In this way, if we start with a decomposition of depth $2^{\OO(k)} \log n$, we manage to rotate the at most $3$ specified leaves to become children of the root with at most $2^{\OO(k)} \log n$ contraction and splitting operations, taking in total $2^{\OO(k)} \log n$ time.
It can also be shown, by using the ``telescoping'' of the potentials on the paths from these leaves to the root, that this increases the potential $\pot(\Tc)$ by at most $2^{\OO(k)} \log n$.
The implementation of the edge insertion/deletion operation then finishes by adding or deleting a child of the root, and then running the balancing procedure with a queue initialized to contain all of the at most $2^{\OO(k)} \log n$ nodes affected by this process.
This finishes the description of how we maintain a downwards well-linked superbranch decomposition of $G$ of degree at most $2^{\OO(k)}$ and depth at most $2^{\OO(k)} \log n$ with $2^{\OO(k)} \log n$ amortized update time.

\paragraph{Constant-approximation of treewidth.}
As discussed earlier, the superbranch decomposition $\Tc$ that we are maintaining can be directly translated to a tree decomposition by associating each node $t$ with a bag $\bag(t) = V(\torso(t))$.
While the adhesions are guaranteed to have size $\OO(\tw(G))$, the degree of $\Tc$ can be $2^{\OO(k)}$, resulting in a tree decomposition of width $2^{\OO(k)}$.
To decrease the resulting width to $\OO(\tw(G))$, we add a wrapper around all manipulations to the superbranch decomposition $\Tc$, which always replaces each node $t$ by a tree decomposition of the \emph{primal graph} of $\torso(t)$.
The primal graph of $\torso(t)$ is the graph $\primal(\torso(t))$, having vertex set $V(\torso(t))$, and an edge between two vertices $u$ and $v$ whenever there is a hyperedge $e \in E(\torso(t))$ containing both $u$ and $v$.
In particular, for each $e \in E(\torso(t))$ the set $V(e)$ is a clique in $\primal(\torso(t))$.

By using \Cref{lem:ovwllem}, i.e., the correspondence between well-linked sets in $\torso(t)$ and $G$, we obtain an algorithm that computes a tree decomposition $\Tc_t$ of $\primal(\torso(t))$ of width at most $9 \cdot \tw(G) + 8$, in time $2^{\OO(k)} \|\torso(t)\|^{\OO(1)} = 2^{\OO(k)}$.
We maintain such tree decompositions for all nodes $t$ of $\Tc$.
Because each adhesion $\adh(st)$ of an edge $st$ incident to $t$ is a clique in $\primal(\torso(t))$, these tree decompositions of torsos can be glued together by following the structure of $\Tc$, to obtain a tree decomposition $\tilde{\Tc}$ of $G$ of width at most $9 \cdot \tw(G) + 8$.
Furthermore, as these tree decompositions have size at most $2^{\OO(k)}$ each, the depth of $\tilde{\Tc}$ is $2^{\OO(k)} \log n$.
Also, it is not hard to simultaneously make $\tilde{\Tc}$ a binary tree decomposition by making each of the decompositions $\Tc_t$ a binary tree with distinct leaves for each child.

\paragraph{Maintaining dynamic programming schemes.}
The applications of dynamic treewidth require maintaining bottom-up dynamic programming schemes on the tree decomposition.
We follow the approach of~\cite{DBLP:conf/focs/KorhonenMNP023} for formalizing this via \emph{tree decomposition automata} and \emph{prefix-rebuilding updates}.
A tree decomposition automaton is an automaton operating on a rooted binary tree decomposition, where the state of a node is computed based on the states of its children, the bag of the node, the edges in the bag of the node, and the bags of its children.
The \emph{evaluation time} of such an automaton is the running time for computing a state of a node based on this information.
Most of the typical dynamic programming schemes on tree decompositions, with overall running time $\evaltime(k) \cdot n$ for some $\evaltime(k)$, can be interpreted as tree decomposition automata with evaluation time $\evaltime(k)$.

A \emph{prefix} of a rooted tree is a connected subtree that contains the root.
A prefix-rebuilding update of a tree decomposition updates it by deleting a prefix of the tree decomposition, and replacing it by a different prefix without changing the subtrees below.
It is easy to see that if a tree decomposition is updated by a prefix-rebuilding update, then we need to re-compute the tree automata states only for the new prefix.
In particular, this can be done in time linear in the size of the new prefix times the evaluation time $\evaltime(k)$ of the automaton.

It remains to argue that the manipulations to the tree decomposition made by our algorithm can be phrased in terms of prefix-rebuilding updates.
This would be trivial if we allowed an extra $2^{\OO(k)} \log n$ factor overhead from the depth of the decomposition, but it can also be done without any significant overhead.
The key idea is to just group all of the local changes to the tree decomposition that are caused by a single update to the graph $G$ into a single prefix-rebuilding update.
It turns out that our algorithm already has the property that the sequence of updates to the tree decomposition caused by a single update to $G$ touches essentially all nodes in some prefix of the tree decomposition, and therefore can be phrased as a prefix-rebuilding operation without significant overhead.
It follows that we can maintain the states of a tree decomposition automaton with evaluation time $\evaltime(k)$ within amortized update time $\evaltime(k) \cdot 2^{\OO(k)} \log n$.

%% file: prelim.tex
\section{Preliminaries}
\label{sec:preli}
We introduce our definitions and state some preliminary results.
None of the definitions here are new, most of them are standard and some of them are from~\cite{DBLP:journals/corr/abs-2411-02658}, which are in turn based on the work of Robertson and Seymour (e.g.~\cite{RobertsonS91}).

\paragraph{Miscellaneous.}
For a function $f \colon X \to Y$ and a set $Z \subseteq X$, we denote by $\funrestriction{f}{Z} \colon Z \to Y$ the restriction of $f$ to $Z$.
For a set $S$, we denote by $\binom{S}{2}$ the set of all unordered pairs of elements from $S$.
For integers $a$ and $b$, we denote by $[a,b]$ the set of integers $i$ with $a \le i \le b$.
We use $[n]$ as a shorthand for $[1,n]$.
Logarithms are base-$2$ unless stated otherwise.
In the context of graph-theoretical notation, we may include the graph $G$ in the subscript to clarify which graph we are talking about.

We assume the standard model of computation in the context of graph algorithms, i.e., the word RAM model with $\Theta(\log n)$-bit words, but we do not abuse this in any way, i.e., we do not use any bit-tricks.

\paragraph{Graphs.}
A graph $G$ consists of a set of vertices $V(G)$ and a set of edges $E(G) \subseteq \binom{V(G)}{2}$.
The \emph{size} of a graph $G$ is $\|G\| = |V(G)| + |E(G)|$.
\emph{Contracting} an edge $uv$ in a graph $G$ is the operation of replacing the two vertices $u$ and $v$ by a single vertex $w_{uv}$ that is adjacent to all vertices that were adjacent to at least one of $u$ or $v$.

\paragraph{Trees.}
A tree is an acyclic connected graph.
To better distinguish trees from graphs, we sometimes call the vertices of a tree \emph{nodes}.
We define that a node of a tree is \emph{leaf} if its degree is $\le 1$.
The set of leaves of a tree $T$ is denoted by $\leaves(T)$.
A node that is not a leaf is an \emph{internal node}.
The set of internal nodes is denoted by $\vint(T) = V(T) \setminus \leaves(T)$.

A rooted tree is a tree $T$ where one node has been chosen as the root.
We remark that the root may be a leaf, and this turns out to be technically convenient for us and we will extensively use rooted trees where the root is a leaf in this paper.
The \emph{parent} of a non-root node $t$ in a rooted tree is the unique neighbor of $t$ on the unique path from $t$ to the root.
The \emph{grandparent} of $t$ (if one exists) is the parent of the parent.

If $T$ is a rooted tree and $t \in V(T)$, we denote by $\ddeg(t)$ the number of children of $t$.
For a set $X \subseteq V(T)$, we denote by $\ddeg(X) = \max_{t \in X} \ddeg(t)$ the maximum number of children of any node in $X$.
We also use $\ddeg(T) = \ddeg(V(T))$.
A rooted tree $T$ is \emph{binary} if $\ddeg(T) \le 2$.
For a node $t \in V(T)$, we denote by $\chd(t)$ the set of children of $t$.
In particular, $\ddeg(t) = |\chd(t)|$.
For a set $X \subseteq V(T)$, we denote $\chd(X) = \bigcup_{t \in X} \chd(t) \setminus X$.

A node $t$ is an \emph{ancestor} of a node $s$ if $t$ is on the unique path from $s$ to the root, and if $t$ is an ancestor of $s$, then $s$ is a \emph{descendant} of $t$.
In particular, every node is both an ancestor and a descendant of itself.
The set of ancestors of a node $t$ is denoted by $\anc(t)$, and for a set $X \subseteq V(T)$ we denote $\anc(X) = \bigcup_{t \in X} \anc(t)$. 
The set of descendants of a node $t$ is denoted by $\desc(t)$.
A \emph{prefix} of a rooted tree $T$ is a set $P \subseteq V(T)$ so that $P = \anc(P)$.
In other words, a prefix is a connected set of nodes that contains the root.
For a node $t$, we denote by $\leaves[t] = \desc(t) \cap \leaves(T)$ the descendants of $t$ that are leaves.

The \emph{depth} of a node $t$ of a rooted tree, denoted by $\depth(t)$ is the number of edges on the unique path between $t$ and the root.
In particular, the depth of the root is $0$.
The \emph{depth} of a rooted tree $T$, denoted by $\depth(T)$, is the maximum depth of its nodes.

\paragraph{Tree decompositions.}
A \emph{tree decomposition} of a graph $G$ is a pair $\Tc = (T,\bag)$, where $T$ is a tree and $\bag \colon V(T) \to 2^{V(G)}$ is a function mapping each node of $T$ to a \emph{bag} of vertices, that satisfies
\begin{enumerate}
\item $V(G) = \bigcup_{t \in V(T)} \bag(t)$,
\item $E(G) \subseteq \bigcup_{t \in V(T)} \binom{\bag(t)}{2}$, and
\item for each $v \in V(G)$, the set $\{t \in V(T) \mid v \in \bag(t)\}$ induces a connected subtree of $T$.
\end{enumerate}
The \emph{width} of a tree decomposition $\Tc$ is the maximum size of a bag minus one, and is denoted by $\width(\Tc)$.
The \emph{treewidth} of a graph $G$ is the minimum width of a tree decomposition of $G$ and is denoted by $\tw(G)$.
A \emph{rooted tree decomposition} is a tree decomposition where the tree $T$ is a rooted tree.
The \emph{size} of a tree decomposition $\Tc$ is $\|\Tc\| = \|T\|+\sum_{t \in V(T)} |\bag(t)|$.

\paragraph{Hypergraphs.}
Instead of graphs, in most of the technical sections of this paper we work with hypergraphs.
A hypergraph $G$ consists of a set of vertices $V(G)$, a set of hyperedges $E(G)$, and a mapping $V \colon E(G) \to 2^{V(G)}$ that associates each hyperedge with a set of vertices.
There may be distinct hyperedges $e_1, e_2 \in E(G)$ so that $V(e_1) = V(e_2)$.
For a set of hyperedges $A \subseteq E(G)$, we denote by $V(A) = \bigcup_{e \in A} V(e)$ the union of the vertices in the hyperedges.
We require that all hypergraphs satisfy $V(G) = V(E(G)) = \bigcup_{e \in E(G)} V(e)$.
The \emph{size} of a hypergraph is $\|G\| = |V(G)| + \sum_{e \in E(G)} (|V(e)|+1)$.

For a hypergraph $G$, the \emph{primal graph} of $G$ is the graph $\primal(G)$ with $V(\primal(G)) = V(G)$ and $E(\primal(G))$ containing an edge $uv$ whenever there is $e \in E(G)$ with $u,v \in V(e)$.
For a graph $G$, the \emph{support hypergraph} of $G$ is the hypergraph $\su{G}$ with $V(\su{G}) = V(G)$, and $E(\su{G})$ containing
\begin{itemize}
\item the hyperedge $e_{\bot}$ with $V(e_{\bot}) = \emptyset$,
\item for each $v \in V(G)$, the hyperedge $e_v$ with $V(e_v) = \{v\}$, and
\item for each $uv \in E(G)$, the hyperedge $e_{uv}$ with $V(e_{uv}) = \{u,v\}$.
\end{itemize}

Note that the size of $\su{G}$ is up to a constant factor the same as the size of $G$.
Also, note that for every graph $G$, $\primal(\su{G}) = G$.

A \emph{separation} of a hypergraph $G$ is a bipartition $(A,B)$ of $E(G)$, i.e., a pair of subsets $A,B \subseteq E(G)$ so that $A \cap B = \emptyset$ and $A \cup B = E(G)$.
The \emph{order} of a separation $(A,B)$ is $|V(A) \cap V(B)|$.
For a set $A \subseteq E(G)$, we denote by $\co{A}$ the \emph{complement} of $A$, i.e., $\co{A} = E(G) \setminus A$.
We denote by $\bd(A) = V(A) \cap V(\co{A})$ the \emph{boundary} of $A$, and by $\bdc(A) = |\bd(A)|$ the size of the boundary of $A$.
Note that the order of a separation $(A,B) = (A,\co{A})$ is $\bdc(A) = \bdc(B) = \bdc(\co{A})$.

Now the \emph{submodularity} of separators can be expressed in a clean way, in particular, the function $\bdc \colon 2^{E(G)} \to \mathbb{Z}_{\ge 0}$ is a \emph{symmetric submodular function}, meaning
\begin{itemize}
\item $\bdc(A \cup B) + \bdc(A \cap B) \le \bdc(A) + \bdc(B)$ for all $A,B \subseteq E(G)$ (submodularity), and
\item $\bdc(A) = \bdc(\co{A})$ for all $A \subseteq E(G)$ (symmetry).
\end{itemize}
The proof of this can be found for example in~\cite{RobertsonS91}.
The symmetry and submodularity of $\bdc$ is our main graph-theoretical tool.

A hypergraph is \emph{normal} if each of its vertices appears in at least two hyperedges.
In particular, for a normal hypergraph $G$ and any $e \in E(G)$, it holds that $\bd(\{e\}) = V(e)$.

For a hypergraph $G$ and a set $A \subseteq E(G)$, we define $G \rescliqs A$ to be the hypergraph with vertex set $V(G \rescliqs A) = V(\co{A})$ and edge set $E(G \rescliqs A) = \co{A} \cup \{e_A\}$, where $V_{G \rescliqs A}(e) = V_G(e)$ for all $e \in \co{A}$, and $V_{G \rescliqs A}(e_A) = \bd(A)$.
In particular, this replaces the set of hyperedges $A$ with a single hyperedge $e_A$ consisting of the boundary of $A$.
We observe that sets of hyperedges in $G \rescliqs A$ can be mapped to sets of hyperedges in $G$.
In particular, for a set $B \subseteq E(G \rescliqs A)$, we denote

\[B \orescliqs A =\begin{cases}
B & \text{ if $e_A \notin B$}\\
B \setminus \{e_A\} \cup A & \text{ if $e_A \in B$}.
\end{cases}\]
This mapping has the nice property that $\bd_G(B \orescliqs A) = \bd_{G \rescliqs A}(B)$.
It follows that for a separation $(B,\co{B})$ of $G \rescliqs A$, $(B \orescliqs A, \co{B} \orescliqs A)$ is a separation of $G$ of the same order.

\paragraph{Well-linked sets.}
Let $G$ be a hypergraph.
A set $A \subseteq E(G)$ is \emph{well-linked} if for all bipartitions $(C_1, C_2)$ of $A$, it holds that either $\bdc(C_1) \ge \bdc(A)$ or $\bdc(C_2) \ge \bdc(A)$.
Well-linkedness will be the core graph-theoretical concept in this paper.
The \emph{well-linked-number} of a hypergraph $G$, denoted by $\wl(G)$, is the largest integer $k$ so that there is a well-linked set $A \subseteq E(G)$ with $\bdc(A) = k$.
For a hyperedge $e \in E(G)$, we also denote by $\wl_e(G)$ the largest integer $k$ so that there is a well-linked set $A \subseteq E(G) \setminus \{e\}$ with $\bdc(A) = k$.

Note that every set $A \subseteq E(G)$ with $|A| \le 1$ is well-linked.
Furthermore, it can be shown that every set $A$ with $\bdc(A) \le 1$ is well-linked.
In particular, $E(G)$ is always well-linked.

We will need an algorithm that tests if a set of hyperedges in a hypergraph is well-linked, and if not, outputs a bipartition witnessing it.
Such an algorithm follows from well-known techniques~\cite{RobertsonS-GMXIII} (see also~\cite[Section~7.6]{DBLP:books/sp/CyganFKLMPPS15}), but we also present a proof in \Cref{sec:app:missingproofs} using our notation.\footnote{We mark by (\inapp) the lemmas whose proofs are presented in \Cref{sec:app:missingproofs}.}

\begin{restatable}[{\cite{RobertsonS-GMXIII}}, \inapp]{lemma}{lemtestwl}
\label{lem:testwl}
There is an algorithm that, given a hypergraph $G$ and a set $A \subseteq E(G)$, in time $2^{\OO(\bdc(A))} \cdot \|G\|^{\OO(1)}$ either
\begin{itemize}
\item returns a bipartition $(C_1, C_2)$ of $A$ so that $\bdc(C_i) < \bdc(A)$ for both $i \in [2]$, or
\item concludes that $A$ is well-linked.
\end{itemize}
\end{restatable}

\paragraph{Superbranch decompositions.}
A superbranch decomposition of a hypergraph $G$ is a pair $\Tc = (T,\lmap)$, where $T$ is a tree whose every internal node has degree $\ge 3$ and $\lmap \colon \leaves(T) \to E(G)$ is a bijection from the leaves of $T$ to $E(G)$.

For an edge $uv \in E(T)$, we denote by $\lmap(\vec{uv}) \subseteq E(G)$ the hyperedges of $G$ that correspond to leaves that closer to $u$ than $v$.
In particular, $\lmap(\vec{uv})$ consists of the hyperedges $e$ so that the unique path from $\lmap^{-1}(e)$ to $v$ contains $u$.
Note that $(\lmap(\vec{uv}), \lmap(\vec{vu}))$ is a separation of $G$, and we say that such a separation is a separation of $\Tc$.
The \emph{adhesion} at an edge $uv \in E(T)$ is the set $\adh(uv) = \bd(\lmap(\vec{uv})) = \bd(\lmap(\vec{vu}))$.
We denote the maximum size of an adhesion of $\Tc$ by $\adhsize(\Tc)$.

The \emph{torso} of an internal node $t \in \vint(T)$ of a superbranch decomposition $\Tc = (T,\lmap)$ is the hypergraph $\torso(t)$ with
\begin{itemize}
\item $E(\torso(t)) = \{e_s \mid st \in E(T)\}$,
\item $V(e_s) = \adh(st)$ for each $e_s \in E(\torso(t))$, and
\item $V(\torso(t)) = \bigcup_{e_s \in E(\torso(t))} V(e_s)$.
\end{itemize}
In particular, $\torso(t)$ is the hypergraph obtained by repeatedly applying the $\rescliqs$ operation as
\begin{align}
\label{ali:torsodef}
\torso(t) = G \rescliqs \lmap(\vec{s_1 t}) \rescliqs \lmap(\vec{s_2 t}) \rescliqs \ldots \rescliqs \lmap(\vec{s_\ell t}),
\end{align}
where $s_1, s_2, \ldots, s_\ell$ are the neighbors of $t$ in $T$.
Note that the number of hyperedges of $\torso(t)$ is the number of neighbors of $t$, and the sizes of the hyperedges of $\torso(t)$ are bounded by $\adhsize(\Tc)$.
With \Cref{ali:torsodef} in mind, for a set $A \subseteq E(\torso(t))$, we denote $A \orescliqs \Tc = \bigcup_{e_s \in A} \lmap(\vec{s t})$.
We observe the following connection between tree decompositions and superbranch decompositions.

\begin{observation}
\label{obs:superbdtotd}
If $\Tc = (T,\lmap)$ is a superbranch decomposition of a hypergraph $G$, then $(T,\bag)$, where $\bag(\ell) = V(\lmap(\ell))$ for $\ell \in \leaves(T)$ and $\bag(t) = V(\torso(t))$ for $t \in \vint(T)$, is a tree decomposition of $\primal(G)$.
\end{observation}

A \emph{rooted superbranch decomposition} is a superbranch decomposition $(T,\lmap)$ where $T$ is a rooted tree.
For a hypergraph $G$ and $e \in E(G)$, an $e$-rooted superbranch decomposition of $G$ is a rooted superbranch decomposition where the root is the node $\lmap^{-1}(e)$.
In this paper, we will maintain an $e_{\bot}$-rooted superbranch decomposition of the support hypergraph $\su{G}$ of the dynamic graph $G$.
Using a superbranch decomposition rooted at a leaf has the technical advantage that every internal node has a parent, reducing the number of cases to consider.

In a rooted superbranch decomposition, for a node $t \in V(T)$ we denote by $\lmap[t] \subseteq E(G)$ the set of hyperedges of $G$ that are mapped to leaves in $\leaves[t]$.

\paragraph{Representation of objects.}
We assume that graphs are represented in the adjacency list format, where edges can be inserted in $\OO(1)$ time, and deleted, given a pointer to the edge, in $\OO(1)$ time.
Note that this does not allow querying if there is an edge between $u$ and $v$ in $\OO(1)$ time.
Hypergraphs are represented as bipartite graphs, where one side of the bipartition is $V(G)$ and other is $E(G)$, and there is an edge between $v \in V(G)$ and $e \in E(G)$ if $v \in V(e)$.
A tree is represented as a graph, and a rooted tree as a tree that contains an additional global pointer pointing to the root node, and for each non-root node $t$ a pointer pointing to the edge $tp$ between $t$ and its parent $p$.
A representation of a tree decomposition $(T,\bag)$ consists of a representation of $T$ and a representation of $\bag$ where each $\bag(t)$ is represented as a linked list to which $t$ contains a pointer to.

A representation of a superbranch decomposition $\Tc = (T,\lmap)$ consists of a representation of $T$, and additionally,
\begin{itemize}
\item $\lmap \colon \leaves(T) \to E(G)$ represented as each leaf storing a pointer to the corresponding node,
\item the inverse $\lmap^{-1} \colon E(G) \to \leaves(T)$ represented as each $e \in E(G)$ containing a pointer to the corresponding leaf,
\item for each edge $st \in E(T)$, the set $\adh(st)$,
\item for each internal node $t \in \vint(T)$, the hypergraph $\torso(t)$,
\item for each hyperedge $e_s \in \torso(t)$, a pointer to the corresponding edge $st$ of $T$, and from each edge $st \in E(T)$, pointers to the corresponding hyperedges $e_t$ of $\torso(s)$ and $e_s$ of $\torso(t)$, and
\item for each node $t \in V(T)$, the number $|\leaves[t]|$ of leaf descendants of it.
\end{itemize}

%% file: gt.tex
\section{Downwards well-linked superbranch decompositions}
\label{sec:downwl}
In this section we define \emph{downwards well-linked superbranch decompositions} and discuss their properties.
In our algorithm, we will maintain a downwards well-linked superbranch decomposition of the hypergraph $\su{G}$, where $G$ is the input dynamic graph.

We define that a rooted superbranch decomposition $\Tc = (T,\lmap)$ of a hypergraph $G$ is \emph{downwards well-linked} if for every node $t \in V(T)$, the set $\lmap[t] \subseteq E(G)$ is well-linked in $G$.
As $\adh(tp) = \bd(\lmap[t])$ for each node $t$ with parent $p$, this implies that $\adhsize(\Tc) \le \wl(G)$.
This connects to treewidth via the following well-known lemma.

\begin{restatable}[{\cite{RobertsonS-GMXIII}}, \inapp]{lemma}{lemwltotwlink}
\label{lem:wltotwlink}
For every graph $G$, $\wl(\su{G}) \le 3 \cdot (\tw(G) + 1)$.
\end{restatable}

Moreover, a converse $\tw(G) \le \wl(\su{G})$ also holds~\cite{Reed97}, but we will not directly use that statement in this paper.

An important property of downwards well-linkedness is that it can be certified in a ``local'' manner.
This will be made formal in \Cref{lem:wltransindecomp}, but to prove it the main tool is the following lemma from~\cite{DBLP:journals/corr/abs-2411-02658}.
We present its proof also here because it is the most important graph-theoretical statement used for our data structure.

\begin{lemma}[{\cite[Lemma~6.3]{DBLP:journals/corr/abs-2411-02658}}]
\label{lem:wltransitive}
Let $G$ be a hypergraph, $A \subseteq E(G)$ a well-linked set, and $B \subseteq E(G \rescliqs A)$.
Then, $B \orescliqs A$ is well-linked in $G$ if and only if $B$ is well-linked in $G \rescliqs A$.
\end{lemma}
\begin{proof}
We prove the only-if-direction first.
Suppose that $B$ is not well-linked in $G \rescliqs A$, and let $(C_1,C_2)$ be a bipartition of $B$ with $\bdc_{G \rescliqs A}(C_i) < \bdc_{G \rescliqs A}(B)$ for both $i \in [2]$.
However, now $(C_1 \orescliqs A, C_2 \orescliqs A)$ is a bipartition of $B \orescliqs A$ with 
\[\bdc_G(C_i \orescliqs A) = \bdc_{G \rescliqs A}(C_i) < \bdc_{G \rescliqs A}(B) = \bdc_G(B \orescliqs A)\]
for both $i \in [2]$, which witnesses that $B \orescliqs A$ is not well-linked in $G$.

We then prove the if-direction.
Let $e_A$ be the hyperedge of $G \rescliqs A$ corresponding to $A$.
Consider first the case that $e_A \notin B$.
Then, $B \orescliqs A = B$, and $\bdc_G(B') = \bdc_{G \rescliqs A}(B')$ for all $B' \subseteq B$, implying that $B \orescliqs A$ is well-linked in $G$ if $B$ is well-linked in $G \rescliqs A$.

Suppose then that $e_A \in B$.
For the sake of contradiction, suppose that $B \orescliqs A$ is not well-linked in $G$, but $B$ is well-linked in $G \rescliqs A$.
There is a bipartition $(C_1, C_2)$ of $B \orescliqs A$ so that $\bdc_G(C_i) < \bdc_G(B \orescliqs A)$ for both $i \in [2]$.
Because $A \subseteq B \orescliqs A$ and $A$ is well-linked, we have that either $\bdc_G(C_1 \cap A) \ge \bdc_G(A)$ or $\bdc_G(C_2 \cap A) \ge \bdc_G(A)$.
Assume without loss of generality that $\bdc_G(C_1 \cap A) \ge \bdc_G(A)$.

We claim that then, the bipartition $(\{e_A\} \cup C_1 \setminus A, C_2 \setminus A)$ of $B$ contradicts that $B$ is well-linked in $G \rescliqs A$.
First,
\begin{align*}
\bdc_{G \rescliqs A}(\{e_A\} \cup C_1 \setminus A) &= \bdc_G(A \cup C_1) && \\
&\le \bdc_G(A) + \bdc_G(C_1) - \bdc_G(A \cap C_1) && \text{(submodularity)}\\
&\le \bdc_G(C_1) && \text{(by $\bdc_G(C_1 \cap A) \ge \bdc_G(A)$)}\\
&< \bdc_G(B \orescliqs A) = \bdc_{G \rescliqs A}(B).
\end{align*}

Second,
\begin{align*}
\bdc_{G \rescliqs A}(C_2 \setminus A) &= \bdc_G(C_2 \cap \co{A}) &&\\
&\le \bdc_G(C_2) + \bdc_G(\co{A}) - \bdc_G(C_2 \cup \co{A}) && \text{(submodularity)}\\
&\le \bdc_G(C_2) + \bdc_G(A) - \bdc_G(C_1 \cap A) && \text{(symmetry)}\\
&\le \bdc_G(C_2) && \text{(by $\bdc_G(C_1 \cap A) \ge \bdc_G(A)$)}\\
&< \bdc_G(B \orescliqs A) = \bdc_{G \rescliqs A}(B).
\end{align*}

Therefore, $B$ is not well-linked in $G \rescliqs A$, which is a contradiction.
\end{proof}

We call the property established by \Cref{lem:wltransitive} the \emph{transitivity} of well-linkedness.
With this, we can prove the following statement, which, informally speaking, asserts that well-linked sets in the torsos of a downwards well-linked superbranch decomposition correspond to well-linked sets in the graph.

\begin{lemma}
\label{lem:wltransindecomp}
Let $G$ be a hypergraph, $e_{\bot} \in E(G)$, and $\Tc = (T,\lmap)$ an $e_{\bot}$-rooted superbranch decomposition of $G$.
Let also $t \in \vint(T)$ be a node with parent $p$, so that $\lmap[c]$ is well-linked for every child $c$ of $t$.
Let $e_p \in E(\torso(t))$ be the hyperedge of $\torso(t)$ corresponding to $p$.
Then, a set $A \subseteq E(\torso(t)) \setminus \{e_p\}$ is well-linked in $\torso(t)$ if and only if $A \orescliqs \Tc$ is well-linked in $G$.
\end{lemma}
\begin{proof}
Recall that $\torso(t) = G \rescliqs \lmap[c_1] \rescliqs \lmap[c_2] \rescliqs \ldots \rescliqs \lmap[c_\ell] \rescliqs \lmap(\vec{pt})$ and $A \orescliqs \Tc = A \orescliqs \lmap[c_1] \orescliqs \lmap[c_2] \orescliqs \ldots \orescliqs \lmap[c_\ell] \orescliqs \lmap(\vec{pt})$, where $c_1, \ldots, c_\ell$ is an enumeration of the children of $t$.

Denote $G' = G \rescliqs \lmap(\vec{pt})$ and $A' = A \orescliqs \lmap[c_1] \orescliqs \lmap[c_2] \orescliqs \ldots \orescliqs \lmap[c_\ell] \subseteq E(G')$.
Because each $\lmap[c_i]$ is well-linked, we can repeatedly apply \Cref{lem:wltransitive} to conclude that $A$ is well-linked in $\torso(t)$ if and only if $A'$ is well-linked in $G'$.
Now, to conclude that $A'$ is well-linked in $G'$ if and only if $A \orescliqs \Tc = A' \orescliqs \lmap(\vec{pt})$ is well-linked in $G$, it suffices to observe that $A'$ does not contain the hyperedge $e_p$ corresponding to the set $\lmap(\vec{pt})$, and therefore $A' = A \orescliqs \Tc$ and for all subsets $A'' \subseteq A'$ it holds that $\bdc_{G'}(A'') = \bdc_{G}(A'')$.
\end{proof}

In particular, \Cref{lem:wltransindecomp} implies that $\Tc$ is downwards well-linked if and only if, for each $t \in \vint(T)$, the set $E(\torso(t)) \setminus \{e_p\}$ is well-linked in $\torso(t)$.
(For the root $r$, we have that $\lmap[r] = E(G)$, which is always well-linked.)
It also implies that in a downwards well-linked superbranch decomposition, for each $t \in \vint(T)$, we have $\wl_{e_p}(\torso(t)) \le \wl(G)$.

In \Cref{lem:wltransindecomp} we assumed only that $\lmap[c]$ is well-linked for each child $c$ of $t$.
This was mostly for illustrative purposes; in our algorithm we will at all times maintain the stronger property that $\Tc$ is downwards well-linked.

%% file: rots.tex
\section{Manipulating superbranch decompositions}
\label{sec:mani}
In this section we introduce our framework of \emph{sequences of basic rotations} for describing updates to superbranch decompositions.
A basic rotation is a local modification concerning only one or two nodes of the superbranch decomposition.
We also give higher-level primitives for manipulating downwards well-linked superbranch decompositions via sequences of basic rotations, which will then be further applied in the subsequent sections.

\subsection{Basic rotations}
There are four basic rotations: splitting a node, contracting an edge, inserting a leaf, and deleting a leaf.
All manipulations to superbranch decompositions will be done via these operations.
Splitting a node and contracting an edge are reverses of each other, as are obviously inserting and deleting a leaf.
In what follows, let $G$ be a hypergraph, $e_{\bot} \in E(G)$, and $\Tc = (T,\lmap)$ an $e_{\bot}$-rooted superbranch decomposition of $G$.

\paragraph{Splitting.}
Let $t \in \vint(T)$ be an internal node of $T$ and $(C, \co{C})$ a separation of $\torso(t)$ with $|C|,|\co{C}| \ge 2$.
\emph{Splitting} $t$ with $(C,\co{C})$ means replacing $t$ by two nodes, $t_{C}$ and $t_{\co{C}}$, so that $t_C$ is adjacent to each neighbor $s$ of $t$ with $e_s \in C$, $t_{\co{C}}$ is adjacent to each neighbor $s$ of $t$ with $e_s \in \co{C}$, and $t_C$ and $t_{\co{C}}$ are adjacent to each other.
Note that $\torso(t_C) = \torso(t) \rescliqs \co{C}$ and $\torso(t_{\co{C}}) = \torso(t) \rescliqs C$.
We observe that a representation of $\Tc$ can be turned into a representation of the superbranch decomposition resulting from splitting $t$ with $(C, \co{C})$ in time $\OO(\|\torso(t)\|)$.

\paragraph{Contracting.}
Let $st \in E(T)$ be an edge of $T$ so that $s,t \in \vint(T)$.
\emph{Contracting} $st$ means simply contracting the edge $st$ of $T$, while keeping the mapping $\lmap$ the same.
We observe that a representation of $\Tc$ can be turned into a representation of $\Tc$ with $st$ contracted in time $\OO(\|\torso(s)\|+\|\torso(t)\|)$.

\paragraph{Inserting a leaf.}
Let $t \in \vint(T)$ be an internal node of $T$, and denote by $\cl(t) = \leaves(T) \cap \chd(t)$ the leaves of $T$ that are children of $t$, and recall that $V(\lmap(\cl(t)))$ is the set of vertices of $G$ in the hyperedges associated with those leaves.
Now, for $X \subseteq V(\lmap(\cl(t)))$, \emph{inserting} $X$ as a child of $t$ means adding a hyperedge $e_X$ with $V(e_X) = X$ to $G$, adding a leaf-node $\ell_X$ as a child of $t$ in $T$, and setting $\lmap(\ell_X) = e_X$.

\begin{lemma}
A representation of $\Tc$ can be turned into a representation of $\Tc$ with $X$ inserted as a child of $t$ in time $\OO(|X| \cdot \|\torso(t)\| + |\anc(t)|)$.
\end{lemma}
\begin{proof}
Denote the new superbranch decomposition by $\Tc'$.
The property that $X \subseteq V(\lmap(\cl(t)))$ guarantees that if $uv$ is an edge of $\Tc'$ that is not between $t$ and a child of $t$, then $\adh_{\Tc'}(uv) = \adh_{\Tc}(uv)$.
Therefore, we only need to update adhesions between $t$ and its children.
The only torso that needs to be updated is the torso of $t$.
Then, we need to increase the stored number of descendant leaves for all ancestors of $t$.
This can be implemented in $\OO(|X| \cdot \|\torso(t)\| + |\anc(t)|)$ time.
\end{proof}

\paragraph{Deleting a leaf.}
Let $t \in \vint(T)$ be an internal node of $T$ that has at least $3$ children.
Let $\ell \in \cl(t)$ so that $V(\lmap(\ell)) \subseteq V(\lmap(\cl(t) \setminus \{\ell\}))$.
\emph{Deleting} $\ell$ means deleting $\lmap(\ell)$ from $G$ and $\ell$ from $T$.

\begin{lemma}
A representation of $\Tc$ can be turned into a representation of $\Tc$ with $\ell$ deleted in time $\OO(|V(\lmap(\ell))| \cdot \|\torso(t)\| + |\anc(t)|)$.
\end{lemma}
\begin{proof}
Denote the new superbranch decomposition by $\Tc'$.
The property that $V(\lmap(\ell)) \subseteq V(\lmap(\cl(t) \setminus \{\ell\}))$ guarantees that if $uv$ is an edge of $\Tc'$ that is not between $t$ and a child of $t$, then $\adh_{\Tc'}(uv) = \adh_{\Tc}(uv)$.
Therefore we only need to recompute adhesions between $t$ and its children, and the only torso to recompute is $\torso(t)$.
Then, we need to decrease the stored numbers of descendant leaves for the ancestors of $t$.
This can be implemented in time $\OO(|V(\lmap(\ell))| \cdot \|\torso(t)\| + |\anc(t)|)$.
\end{proof}

\paragraph{Sequences of basic rotations.}
In our algorithm, we manipulate sequences of basic rotations.
A \emph{sequence $\seq$ of basic rotations} stores for each rotation in the sequence all information necessary to perform it: For splitting, the node $t$ and the bipartition $(C,\co{C})$ of $\torso(t)$ are stored, for contracting, the pair of nodes $s,t$ is stored, for inserting a leaf, the node $t$ and the set $X \subseteq V(G)$ are stored, and for deleting a leaf, the leaf node $\ell$ is stored.

We define the \emph{size} $\|\seq\|$ of a sequence of basic rotations $\seq$ so that the basic rotations in $\seq$ can be performed in time $\OO(\|\seq\|)$.
In particular, for splitting the size is $\|\torso(t)\|$, for contraction the size is $\|\torso(t)\|+\|\torso(s)\|$, for inserting a leaf the size is $|X| \cdot \|\torso(t)\| + |\anc(t)|$, and for deleting a leaf the size is $|V(\lmap(\ell))| \cdot \|\torso(t)\| + |\anc(t)|$.
Then, the size $\|\seq\|$ of $\seq$ is the sum of the sizes of the basic rotations in it.
We assume that $\seq$ is stored as a linked list, in particular, so that we can append and prepend basic rotations to $\seq$ efficiently, i.e., in time linear in the size of the appended or prepended basic rotations.

Let $\seq$ be a sequence of basic rotations that transforms $\Tc = (T,\lmap)$ into $\Tc' = (T',\lmap')$.
We denote by $V_{\Tc}(\seq) \subseteq V(T)$ the set of nodes of $T$ involved in the rotations in $\seq$, i.e., all internal nodes involved in splittings and contractions, all leaves deleted, and the parents of all leaves deleted and inserted.
Analogously, $V_{\Tc'}(\seq) \subseteq V(T')$ is the set of nodes of $T'$ involved in $\seq$.
The \emph{trace} of $\seq$ in $\Tc$ is the set $\trace_{\Tc}(\seq)$ of all ancestors of nodes of $T$ involved in $\seq$, i.e., $\trace_{\Tc}(\seq) = \anc_T(V_{\Tc}(\seq))$.
Naturally, $\trace_{\Tc'}(\seq)$ is defined analogously.
We define $\|\seq\|_{\Tc} = \|\seq\| + |\trace_{\Tc}(\seq)|$ to be a size measure of $\seq$ that takes into account traversing the ancestors of $V_{\Tc}(\seq)$.
Note that also $|\trace_{\Tc'}(\seq)| \le \|\seq\|_{\Tc}$ holds.

To cover some corner cases, we allow a sequence of basic rotations $\seq$ to contain also dummy rotations that do not do anything, but just ``touch'' a node in the sense that it will be included in $V_{\Tc}(\seq)$ and $V_{\Tc'}(\seq)$.

\subsection{Splitting a node}
In our algorithm we maintain a rooted superbranch decomposition that is downwards well-linked and has an upper bound on its maximum degree.
The typical way to modify this superbranch decomposition will be to first use the contraction operation to form a node of high degree, and then the splitting operation to split it up into a subtree of a different form than we started with.
Our core idea is that this splitting can be done in a manner that preserves downwards well-linkedness and an upper bound on the degree.
In this subsection we give the subroutine for doing that.

We start with the following algorithm for partitioning any set of hyperedges in a hypergraph into well-linked sets.

\begin{lemma}
\label{lem:partitiontowlalg}
There is an algorithm that, given a hypergraph $G$ and a non-empty set of hyperedges $X \subseteq E(G)$, in time $2^{\OO(\bdc(X))} \cdot \|G\|^{\OO(1)}$ returns a partition $\compset$ of $X$ into at most $|\compset| \le 2^{\bdc(X)}$ sets, so that each $C \in \compset$ is well-linked in $G$.
\end{lemma}
\begin{proof}
We maintain a partition $\compset$ of $X$, initialized to be $\compset = \{X\}$.
We repeatedly apply the algorithm of \Cref{lem:testwl} to test for each part $C \in \compset$ whether $C$ is well-linked, and if not, replace $C$ by the two sets $C_1$, $C_2$ returned by it, where $(C_1,C_2)$ is a bipartition of $C$ with $\bdc(C_i) < \bdc(C)$ for both $i \in [2]$.

We observe that this process maintains that $\sum_{C \in \compset} 2^{\bdc(C)} \le 2^{\bdc(X)}$, but increases $|\compset|$ in each iteration.
Therefore, it must terminate within at most $2^{\bdc(X)}$ iterations, with $|\compset| \le 2^{\bdc(X)}$.
As the algorithm of \Cref{lem:testwl} runs in time $2^{\OO(\bdc(C))} \cdot \|G\|^{\OO(1)}$, and $\bdc(C) \le \bdc(X)$ holds for all $C \in \compset$, the total running time is at most $2^{\OO(\bdc(X))} \cdot \|G\|^{\OO(1)}$.
\end{proof}

We then apply the algorithm of \Cref{lem:partitiontowlalg} to create a subroutine for splitting a node while maintaining downwards well-linkedness.

\begin{lemma}
\label{lem:rotatemainint}
Let $G$ be a hypergraph, $e_{\bot} \in E(G)$, and $\Tc = (T,\lmap)$ an $e_{\bot}$-rooted superbranch decomposition of $G$ that is downwards well-linked.
There is an algorithm that, given an internal node $t \in \vint(T)$ and a set of children $X \subseteq \chd(t)$, with $|X| \ge 1$ and $|\bigcup_{x \in X} \adh(xt)| = \alpha$, either 
\begin{enumerate}[label=(\alph*)]
\item transforms $\Tc$ into $\Tc' = (T',\lmap')$ via a sequence $\seq$ consisting of one splitting rotation so that\label{lem:rotatemainint:casetrans}
\begin{enumerate}[label=\arabic*., ref=\arabic*]
\item $\Tc'$ is downwards well-linked,\label{lem:rotatemainint:prop1}
\item $V_{\Tc}(\seq) = \{t\}$ and $|V_{\Tc'}(\seq)| = 2$,\label{lem:rotatemainint:prop2}
\item all nodes in $X$ are children of the shallowest node of $V_{\Tc'}(\seq)$ in $\Tc'$, and\label{lem:rotatemainint:prop3}
\item for all $t' \in V_{\Tc'}(\seq)$, it holds that $\ddeg_{T'}(t') < \ddeg_{T}(t)$\label{lem:rotatemainint:prop4}, or
\end{enumerate}
\item concludes that $\ddeg_{T}(t) \le |X| + 2^{\wl(G)+\alpha}$.\label{lem:rotatemainint:caseconc}
\end{enumerate}
The running time of the algorithm is $2^{\OO(\wl(G)+\alpha)} \cdot \|\torso(t)\|^{\OO(1)}$ and in case \ref{lem:rotatemainint:casetrans} it returns $\seq$, which has $\|\seq\| \le \OO(\|\torso(t)\|)$.
\end{lemma}
\begin{proof}
Let $e_p \in E(\torso(t))$ be the hyperedge of $\torso(t)$ associated with the parent $p$ of $t$.
Let also $E_X \subseteq E(\torso(t))$ be the set of hyperedges associated with the children of $t$ that are in $X$.
Also, denote $Y = \chd(t) \setminus X$, and let $E_Y \subseteq E(\torso(t))$ be the corresponding hyperedges.
If $Y = \emptyset$, we conclude with the case~\ref{lem:rotatemainint:caseconc}, so assume that $Y$ is non-empty.

Note that $\{\{e_p\}, E_X, E_Y\}$ is a partition of $E(\torso(t))$.
Because $\Tc$ is downwards well-linked, we have $\bdc(e_p) \le |V(e_p)| \le \wl(G)$, and because $|\bigcup_{x \in X} \adh(xt)| = \alpha$ we have $\bdc(E_X) \le \alpha$.
It follows that $\bdc(E_Y) = \bdc(\{e_p\} \cup E_X) \le \wl(G) + \alpha$.

We apply the algorithm of \Cref{lem:partitiontowlalg} to find a partition $\compset$ of $E_Y$ into at most $2^{\wl(G)+\alpha}$ sets, so that each $C \in \compset$ is well-linked in $\torso(t)$.
This runs in time $2^{\OO(\wl(G)+\alpha)} \cdot \|\torso(t)\|^{\OO(1)}$.
If every $C \in \compset$ has size $|C| = 1$, then we conclude that $\ddeg_{T}(t) = |X|+|Y| \le |X| + 2^{\wl(G)+\alpha}$ and return with the case~\ref{lem:rotatemainint:caseconc}.

Otherwise, we take an arbitrary $C \in \compset$ with $|C| \ge 2$, and apply the splitting rotation with the separation $(C,\co{C})$ of $\torso(t)$.
Note that $|\co{C}| \ge 2$ because $e_p \in \co{C}$ and $E_X \subseteq \co{C}$.
This replaces $t$ with two nodes $t_C$ and $t_{\co{C}}$, with $t_C$ adjacent to $t_{\co{C}}$ and each child of $t$ whose corresponding hyperedge is in $C$, and $t_{\co{C}}$ adjacent to $p$, $t_C$, and each child of $t$ whose corresponding hyperedge is in $\co{C}$.

We denote the resulting superbranch decomposition by $\Tc' = (T',\lmap')$, and the sequence consisting of this splitting rotation by $\seq$, and claim that $\Tc'$ and $\seq$ satisfy the required properties.
Note that $V_{\Tc}(\seq) = \{t\}$ and $V_{\Tc'}(\seq) = \{t_C,t_{\co{C}}\}$, so \Cref{lem:rotatemainint:prop2,lem:rotatemainint:prop3} are clear from the construction.
Also, $|C| \ge 2$ implies that $\ddeg_{T'}(t_{\co{C}}) < \ddeg_{T}(t)$, and $|\co{C}| \ge 2$ implies that $\ddeg_{T'}(t_C) < \ddeg_{T}(t)$, so \Cref{lem:rotatemainint:prop4} holds.
We then prove \Cref{lem:rotatemainint:prop1}.

\begin{claim}
$\Tc'$ is downwards well-linked.
\end{claim}
\begin{claimproof}
All edges $xy$ of $T'$, where $y$ is a parent of $x$, except $t_C t_{\co{C}}$, correspond to edges of $T$ in the sense that there is $x'y' \in E(T)$ with $\lmap'(\vec{xy}) = \lmap(x'y')$.
Therefore, it suffices to argue that $\lmap(\vec{t_C t_{\co{C}}})$ is well-linked in $G$, or equivalently, that $C \orescliqs \Tc$ is well-linked in $G$.
Because $\Tc$ is downwards well-linked, $e_p \notin C$, and $C$ is well-linked in $\torso(t)$, this follows from \Cref{lem:wltransindecomp}.
\end{claimproof}

Therefore the algorithm is correct.
The running time and the fact that $\|\seq\| \le \OO(\|\torso(t)\|)$ are also clear from the given arguments.
\end{proof}

We then apply \Cref{lem:rotatemainint} to build a higher-level subroutine for splitting a high-degree node into a subtree with an upper bound on the degree.

\begin{lemma}
\label{lem:rotatemain}
Let $G$ be a hypergraph, $e_{\bot} \in E(G)$, and $\Tc = (T,\lmap)$ an $e_{\bot}$-rooted superbranch decomposition of $G$ that is downwards well-linked.
There is an algorithm that, given an internal node $t \in \vint(T)$ and a set of children $X \subseteq \chd(t)$, with $|\bigcup_{x \in X} \adh(xt)| = \alpha$, transforms $\Tc$ into $\Tc' = (T',\lmap')$ via a sequence $\seq$ of basic rotations so that
\begin{enumerate}
\item $\Tc'$ is downwards well-linked,\label{lem:rotatemain:prop1}
\item $V_{\Tc}(\seq) = \{t\}$,\label{lem:rotatemain:prop2}
\item all nodes in $X$ are children of the shallowest node of $V_{\Tc'}(\seq)$ in $\Tc'$, and\label{lem:rotatemain:prop3}
\item $\ddeg_{T'}(V_{\Tc'}(\seq)) \le \max(|X| + 2^{\wl(G)+\alpha}, 1 + 2^{2 \wl(G)})$.\label{lem:rotatemain:prop4}
\end{enumerate}
The running time of the algorithm is $2^{\OO(\wl(G)+\alpha)} \cdot \|\torso(t)\|^{\OO(1)}$ and it returns $\seq$, which has $\|\seq\| \le \|\torso(t)\|^{\OO(1)}$.
\end{lemma}
\begin{proof}
We describe an iterative procedure that transforms $\Tc$ into $\Tc'$.
We denote the current superbranch decomposition by $\Tc' = (T',\lmap')$, which initially equals $\Tc$, and the current sequence of basic rotations, which transforms $\Tc$ into $\Tc'$, by $\seq$.
We initialize $\seq$ to contain one dummy rotation that touches the node $t$, so that initially $V_{\Tc}(\seq) = V_{\Tc'}(\seq) = \{t\}$.
We also maintain a set $A \subseteq V(T')$ of ``active'' nodes, which we initially set as $A = \{t\}$.
Throughout, the invariants we maintain are that $\Tc'$ and $\seq$ satisfy the properties of \Cref{lem:rotatemain:prop1,lem:rotatemain:prop2,lem:rotatemain:prop3}, and for the set $V_{\Tc'}(\seq) \setminus A$ it holds that $\ddeg_{T'}(V_{\Tc'}(\seq) \setminus A) \le \max(|X| + 2^{\wl(G)+\alpha}, 1 + 2^{2 \wl(G)})$.
Furthermore, we maintain that all nodes in $X$ are children of the shallowest node in $V_{\Tc'}(\seq)$.
In particular, once $A = \emptyset$, all of the required properties are satisfied.

As long as $A$ is non-empty, we pick an arbitrary node $v \in A$.
If $X$ is non-empty and $v$ is the shallowest node in $V_{\Tc'}(\seq)$, we apply \Cref{lem:rotatemainint} with $v$ and $X$.
If it concludes that $\ddeg_{T'}(v) \le |X| + 2^{\wl(G)+\alpha}$, we simply remove $v$ from $A$, which maintains the invariant in this case.
If it splits $v$ into two nodes, then we replace $v$ in $A$ by the two resulting nodes.
Note that by the guarantees of \Cref{lem:rotatemainint}, this also maintains the invariants.

In the other case, no child (or a descendant) of $v$ is in $X$.
In this case, let $c$ be an arbitrary child of $v$.
Because $\Tc'$ is downwards well-linked, we have that $|\adh(cv)| \le \wl(G)$.
We apply \Cref{lem:rotatemainint} with $v$ and the set $\{c\}$.
If it concludes that $\ddeg_{T'}(v) \le 1 + 2^{\wl(G)+|\adh(cv)|} \le 1+2^{2 \wl(G)}$, we can remove $v$ from $A$ while maintaining the invariant.
If it splits $v$ into two nodes, then we replace $v$ in $A$ by the two resulting nodes.
By the guarantees of \Cref{lem:rotatemainint}, this maintains the invariants.

Because \Cref{lem:rotatemainint} only applies the splitting rotation and maintains that $V_{\Tc}(\seq) = \{t\}$, this process can go on for at most $\OO(\ddeg_{T}(t))$ iterations, which is also an upper bound for $|V_{\Tc'}(\seq)|$ and $|A|$.
Because of this, the process can be easily implemented in time $2^{\OO(\wl(G)+\alpha)} \cdot \|\torso(t)\|^{\OO(1)}$, and $\|\seq\|$ is upper bounded by $\|\torso(t)\|^{\OO(1)}$.
\end{proof}

%% file: struct.tex
\section{The data structure}
\label{sec:struct}
In this section, we introduce the structure and invariants of the superbranch decomposition that we maintain in our algorithm.
We start by stating our main lemma regarding the maintenance of a superbranch decomposition.
Then we introduce the internal invariants of the decomposition, and then the potential function we use for the amortized analysis.

The following is the main lemma of this paper.
Its proof spans \Cref{sec:struct,sec:balancing,sec:insdeledge,sec:mainlemmaproof}.

\begin{restatable}{lemma}{mainintlemma}
\label{lem:mainintlemma}
Let $G$ be a dynamic graph and $k \ge 1$ an integer with a promise that $\wl(\su{G}) \le k$ at all times.
There is a data structure that maintains $G$, $\su{G}$, and an $e_{\bot}$-rooted superbranch decomposition $\Tc = (T,\lmap)$ of $\su{G}$ so that 
\begin{itemize}
\item $\Tc$ is downwards well-linked,
\item $\ddeg(T) \le 2^{\OO(k)}$, and
\item $\depth(T) \le 2^{\OO(k)} \log \|G\|$.
\end{itemize}
The data structure supports the following operations:
\begin{itemize}
\item $\mathsf{Init}(G,k)$: Given an edgeless graph $G$ and an integer $k \ge 1$, initialize the data structure with $G$ and $k$, and return $\Tc$.
Runs in $2^{\OO(k)} \|G\|$ amortized time.
\item $\mathsf{AddEdge}(uv)$: Given a new edge $uv \in \binom{V(G)}{2} \setminus E(G)$, add $uv$ into $G$. Runs in $2^{\OO(k)} \log \|G\|$ amortized time.
\item $\mathsf{DeleteEdge}(uv)$: Given an edge $uv \in E(G)$, delete $uv$ from $G$. Runs in $2^{\OO(k)} \log \|G\|$ amortized time.
\end{itemize}
Furthermore, in the operations $\mathsf{AddEdge}$ and $\mathsf{DeleteEdge}$, $\Tc$ is updated by a sequence $\seq$ of basic rotations, which is returned.
The sizes $\|\seq\|_{\Tc}$ of these sequences have the same amortized upper bound as the running time.
\end{restatable}

\paragraph{Good and semigood superbranch decompositions.}
We then define the internal invariants that the superbranch decomposition of \Cref{lem:mainintlemma} will satisfy.
We define \emph{$k$-good} and \emph{$k$-semigood} superbranch decompositions, where $k$-good captures the properties we want to eventually maintain, and $k$-semigood is a relaxation of $k$-good to which we settle between subroutines.

A rooted superbranch decomposition $\Tc = (T,\lmap)$ of a hypergraph $G$ is \emph{$k$-semigood} if
\begin{itemize}
\item $\Tc$ is downwards well-linked, and
\item $\ddeg(T) \le 2^{2k}+1$.
\end{itemize}

The value $2^{2k}+1$ comes from \Cref{lem:rotatemain}.

A node $t \in V(T)$ is \emph{$d$-unbalanced} for an integer $d \ge 1$ if there exists a descendant $s$ of $t$ so that $\depth(s) \ge \depth(t)+d$ and $|\lmap[s]| \ge \frac{2}{3} \cdot |\lmap[t]|$.
A node $t \in V(T)$ is \emph{$d$-balanced} if it is not $d$-unbalanced.
A rooted superbranch decomposition $\Tc$ is $k$-good if it is $k$-semigood and all of its non-root nodes are $2^{2k+1}$-balanced.

\begin{lemma}
\label{lem:gooddepth}
Let $G$ be a hypergraph and $\Tc = (T,\lmap)$ a $k$-good rooted superbranch decomposition of $G$.
Then, $\depth(T) \le 2^{\OO(k)} \log \|G\|$.
\end{lemma}
\begin{proof}
Let $d = 2^{2k+1}$, and suppose $T$ contains a path $P = t_0, \ldots, t_\ell$ with $\ell = d \cdot \lceil 2 \cdot \log |E(G)| \rceil$, where $t_0$ is a child of the root and $t_\ell$ a descendant of $t_0$.
Because $\Tc$ is $d$-balanced, we have for all $i \in [\lceil 2 \cdot \log |E(G)| \rceil]$ that $|\lmap[t_{i \cdot d}]| < \frac{2}{3} \cdot |\lmap[t_{(i-1) \cdot d}]|$.
It follows that $|\lmap[t_\ell]| < |E(G)| \cdot \left(\frac{2}{3}\right)^{2 \log |E(G)|} < 1$, a contradiction.
\end{proof}

\paragraph{The potential function.}
We then introduce the potential function for analyzing the amortized running time of our data structure.
The potential function is similar to the potential function for splay trees~\cite{DBLP:journals/jacm/SleatorT85}, but includes a factor depending on the degree of a node in order to accommodate nodes with more than two children.

Let $\Tc = (T,\lmap)$ be a rooted superbranch decomposition.
We define the potential of a single internal node $t \in \vint(T)$ as
\[\pot_{\Tc}(t) = (\ddeg(t)-1) \cdot \log(|\lmap[t]|).\]
Note that an internal node $t$ has always $\ddeg(t) \ge 2$ and $|\lmap[t]| \ge 2$, so all internal nodes have potential at least $1$.
Then, the potential of $\Tc$ is
\[\pot(\Tc) = \sum_{t \in \vint(T)} \pot_{\Tc}(t).\]
We also denote for a set $X \subseteq \vint(T)$ that $\pot_{\Tc}(X) = \sum_{t \in X} \pot_{\Tc}(t)$.

The motivation for the factor $(\ddeg(t)-1)$ in the potential is that for any connected set $X \subseteq \vint(T)$, it holds that
\[\sum_{t \in X} (\ddeg(t)-1) = |\chd(X)|-1.\]

%% file: balance.tex
\section{Balancing}
\label{sec:balancing}
In this section, we give the subroutine for balancing the superbranch decomposition.
More formally, balancing means turning a $k$-semigood superbranch decomposition into a $k$-good superbranch decomposition, while decreasing the potential and using running time proportional to the potential decrease.
Specifically, this section is dedicated to the proof of the following lemma.

\begin{lemma}
\label{lem:mainbalancinglemma}
Let $G$ be a hypergraph, $e_{\bot} \in E(G)$, and $\Tc = (T,\lmap)$ an $e_{\bot}$-rooted superbranch decomposition of $G$.
Suppose also that $k \ge 1$ is an integer so that $\Tc$ is $k$-semigood and $\wl(G) \le k$.
There is an algorithm that, given $k$ and a prefix $R \subseteq V(T)$ of $T$ so that all $2^{2k+1}$-unbalanced nodes of $T$ are in $R$, transforms $\Tc$ into a $k$-good $e_{\bot}$-rooted superbranch decomposition $\Tc'$ with $\pot(\Tc') \le \pot(\Tc)$ via a sequence $\seq$ of basic rotations, and returns $\seq$.
The running time of the algorithm is $2^{\OO(k)} \cdot (|R| + \pot(\Tc)-\pot(\Tc'))$, which is also an upper bound for $\|\seq\|_{\Tc}$.
\end{lemma}

The goal of this section is to prove \Cref{lem:mainbalancinglemma}.
To avoid repeatedly stating assumptions, for the remainder of the section we assume that we are in the setting where $\Tc$, $G$, $k$, and $R$ have the properties as stated in \Cref{lem:mainbalancinglemma}.

We start with an easy observation about finding whether a node is unbalanced.

\begin{lemma}
\label{lem:testunbalancedness}
There is an algorithm that, given a node $t \in V(T)$, in time $2^{\OO(k)}$ either concludes that $t$ is $2^{2k+1}$-balanced, or returns a descendant $s$ of $t$ so that $\depth(s) = \depth(t) + 2^{2k+1}$ and $|\lmap[s]| \ge \frac{2}{3} \cdot |\lmap[t]|$.
\end{lemma}
\begin{proof}
Such a descendant $s$, if one exists, can always be found by traversing downwards from $t$ by always going to the child with the most leaf descendants.
Because $\Tc$ is $k$-semigood and thus has $\ddeg(T) \le 2^{\OO(k)}$, this can be implemented in time $2^{\OO(k)}$ by using the leaf descendant counters stored in the representation of $\Tc$.
\end{proof}

Then, we give an algorithm for performing one step of our balancing procedure.
It takes one unbalanced node as input, and performs rotations to decrease the potential.

\begin{lemma}
\label{lem:balancingonestep}
There is an algorithm that, given a $2^{2k+1}$-unbalanced non-root node $t \in V(T)$, transforms $\Tc$ into a $k$-semigood $e_{\bot}$-rooted superbranch decomposition $\Tc'$ of $G$ via a sequence $\seq$ of basic rotations, so that
\begin{enumerate}
\item $\pot(\Tc') \le \pot(\Tc)-1$,
\item $\|\seq\| \le 2^{\OO(k)}$, and\label{lem:balancingonestep:item2}
\item $V_{\Tc}(\seq) \subseteq \desc(t)$ is a connected set in $T$ and contains $t$.
\end{enumerate}
The running time of the algorithm is $2^{\OO(k)}$ and it returns $\seq$.
\end{lemma}
\begin{proof}
We first apply \Cref{lem:testunbalancedness} to find a descendant $s$ of $t$ so that $\depth(s) = \depth(t)+ 2^{2k+1}$ and $|\lmap[s]| \ge \frac{2}{3} \cdot |\lmap[t]|$.
By following the parent pointers stored in the representation of $\Tc$, we also find in $2^{\OO(k)}$ time the unique $(t,s)$-path in $T$.
Let $P$ denote this path with $s$ removed.
In particular, $P$ contains $2^{2k+1}-1$ edges and $2^{2k+1}$ nodes.
We obtain a superbranch decomposition $\Tc' = (T',\lmap')$ by contracting all edges on $P$ with the contraction basic rotation.
Because $\Tc$ is $k$-semigood and $|P| \le 2^{\OO(k)}$, this runs in time $2^{\OO(k)}$.

Let $t'$ be the node of $\Tc'$ corresponding to the contracted path.
We have that 
\begin{enumerate}
\item $\Tc'$ is downwards well-linked,\label{lem:balancingonestep:enum1:item1}
\item $\ddeg_{T'}(V(T') \setminus \{t'\}) \le 2^{2k}+1$,\label{lem:balancingonestep:enum1:item2}
\item $2^{2k+1}+1 \le \ddeg_{T'}(t') \le 2^{\OO(k)}$,\label{lem:balancingonestep:enum1:item3}
\item $\lmap'[t'] = \lmap[t]$, and\label{lem:balancingonestep:enum1:item4}
\item $s \in \chd_{T'}(t')$ and $\lmap'[s] = \lmap[s]$.\label{lem:balancingonestep:enum1:item5}
\end{enumerate}

Because $\Tc'$ is downwards well-linked, we have that $\adh(st') \le \wl(G)$.
We apply the operation of \Cref{lem:rotatemain} to the node $t'$ and the set $\{s\}$ of children of $t'$.
It runs in time $2^{\OO(k)}$, and transforms $\Tc'$, via a sequence $\seq^*$ of basic rotations, to a superbranch decomposition $\Tc'' = (T'', \lmap'')$ so that
\begin{enumerate}[resume]
\item $\Tc''$ is downwards well-linked,\label{lem:balancingonestep:enum2:item1}
\item $V_{\Tc'}(\seq^*) = \{t'\}$,\label{lem:balancingonestep:enum2:item2}
\item $s$ is a child of the shallowest node in $V_{\Tc''}(\seq^*)$ in $\Tc''$, and\label{lem:balancingonestep:enum2:item3}
\item $\ddeg_{T''}(V_{\Tc''}(\seq^*)) \le \max(1 + 2^{|\adh(st')|+\wl(G)}, 1 + 2^{2 \wl(G)}) \le 2^{2 \wl(G)} + 1 \le 2^{2k} + 1$.\label{lem:balancingonestep:enum2:item4}
\end{enumerate}
It also returns $\seq^*$, which has $\|\seq^*\| \le 2^{\OO(k)}$.

We immediately note that the combination of \Cref{lem:balancingonestep:enum1:item2,lem:balancingonestep:enum2:item2,lem:balancingonestep:enum2:item4} implies that $\ddeg(T'') \le 2^{2k}+1$, implying with \Cref{lem:balancingonestep:enum2:item1} that $\Tc''$ is $k$-semigood.
We then prove the main claim about the potential of $\Tc''$.

\begin{claim}
$\pot(\Tc'') \le \pot(\Tc) -1$.
\end{claim}
\begin{claimproof}
We observe that $\pot(\Tc'') = \pot(\Tc) - \pot_{\Tc}(V(P)) + \pot_{\Tc''}(V_{\Tc''}(\seq^*))$, so the claim is equivalent to the claim that $\pot_{\Tc''}(V_{\Tc''}(\seq^*)) \le \pot_{\Tc}(V(P)) - 1$.
Let $C = \sum_{x \in V(P)} (\ddeg_{T}(x)-1)$.
Because $s$ is a descendant of all nodes in $V(P)$, we have that
\begin{align*}
\pot_{\Tc}(V(P)) &\ge C \cdot \log(|\lmap[s]|)\\
&\ge C \cdot (\log(|\lmap[t]|)+\log(2/3))\\
&\ge C \cdot \log(|\lmap[t]|) - C \cdot \log(3/2).
\end{align*}

Because $V(P)$ is a connected set of internal nodes, we have that $|\chd_{T}(V(P))| = C+1$.
We also have that $|\chd_{T}(V(P))| = |\chd_{T'}(t')| = |\chd_{T''}(V_{\Tc''}(\seq^*))|$.
Let $t''$ be the shallowest node in $V_{\Tc''}(\seq^*)$, and denote $C_r = \ddeg_{T''}(t'')-1$ and $C_o = \sum_{x \in V_{\Tc''}(\seq^*) \setminus \{t''\}} (\ddeg_{T''}(x)-1)$.
Because $V_{\Tc''}(\seq^*)$ is a connected set of internal nodes, we have that $|\chd_{T''}(V_{\Tc''}(\seq^*))| = C_o + C_r + 1$, and therefore $C_o + C_r = C$.

By $C+1 = |\chd_{T'}(t')|$ we have that $C \ge 2^{2k+1}$.
We also have that $C_r \le 2^{2k} \le C/2$, so $C_o \ge C/2$.
Because $s$ is a child of $t''$ in $\Tc''$ and 
\[|\lmap''[s]| = |\lmap[s]| \ge \frac{2}{3} \cdot |\lmap[t]| = \frac{2}{3} \cdot |\lmap''[t'']|,\]
we have
\begin{align*}
\pot_{\Tc''}(V_{\Tc''}(\seq^*)) &\le C_r \cdot \log(|\lmap[t]|) + C_o \cdot \log(|\lmap[t]|-|\lmap[s]|) && \\
&\le C_r \cdot \log(|\lmap[t]|) + C_o \cdot (\log(|\lmap[t]|)-\log(3)) && \text{(by $|\lmap[s]| \ge \frac{2}{3} |\lmap[t]|$)}\\
&\le C \cdot \log(|\lmap[t]|) - C_o \cdot \log(3) && \text{(by $C_o + C_r = C$)}\\
&\le C \cdot \log(|\lmap[t]|) - C \cdot \log(3)/2 && \text{(by $C_o \ge C/2$)}\\
&\le C \cdot \log(|\lmap[t]|) - C \cdot \log(3/2) - 1 && \text{(by $C \ge 2^{2k+1} \ge 8$)}\\
&\le \pot_{\Tc}(V(P))-1. &&
\end{align*}
\end{claimproof}

By prepending a sequence of basic rotations describing the contraction of $P$ to the sequence $\seq^*$, we obtain a sequence $\seq$ of basic rotations that transforms $\Tc$ into $\Tc''$.
We have that $\|\seq\| \le \|\seq^*\| + 2^{\OO(k)} \cdot |V(P)| \le 2^{\OO(k)}$, and $V_{\Tc}(\seq) = V(P) \subseteq \desc(t)$.
We return $\seq$.
\end{proof}

We then finish the proof of \Cref{lem:mainbalancinglemma} by giving an algorithm that repeatedly applies the operation of \Cref{lem:balancingonestep}.

\begin{proof}[Proof of \Cref{lem:mainbalancinglemma}]
We implement an iterative process that improves $\Tc$ by repeatedly applying \Cref{lem:balancingonestep}, and maintains a prefix $R \subseteq V(T)$ so that all $2^{2k+1}$-unbalanced nodes of $\Tc$ are in $R$.
Throughout this process, $R$ will be stored as a stack having the property that if a node $t$ is at a certain position of the stack, all of its ancestors are below $t$, i.e., will be popped after $t$.
From the initial input $R$, such a stack representation can be constructed in $2^{\OO(k)} \cdot |R|$ time by performing a depth-first search that starts from the root and is restricted to nodes in $R$.

We then state all the information and invariants that are maintained in this process.
Let $\Tc_0 = (T_0,\lmap_0)$ and $R_0$ denote the initial superbranch decomposition $\Tc$ and the initial set $R$.
Let also $c \ge 1$ be a constant so that the $2^{\OO(k)}$ factor in \Cref{lem:balancingonestep:item2} of \Cref{lem:balancingonestep} is bounded by $2^{ck}$.
Let $\seq$ denote the sequence of basic rotations applied so far, transforming $\Tc_0$ into $\Tc$.
We will also maintain a set $R_{+} \subseteq V(T_0)$ consisting of all nodes of $\Tc_0$ ``touched'' by the algorithm during the process.
The set $R_{+}$ is not explicitly maintained by the algorithm, but only used for the analysis.
Initially $R_{+}$ equals $R_0$.
In addition to maintaining that $\Tc$ is $k$-semigood and $R$ contains all $2^{2k+1}$-unbalanced nodes of $T$, we will maintain that
\begin{enumerate}
\item $\pot(\Tc) \le \pot(\Tc_0)$,\label{lem:mainbalancinglemma:item1}
\item $\|\seq\| \le 2^{ck} \cdot (\pot(\Tc_0)-\pot(\Tc))$,\label{lem:mainbalancinglemma:item2}
\item $R_{+} = R_0 \cup \trace_{\Tc_0}(\seq)$, and\label{lem:mainbalancinglemma:item3}
\item $|R_{+}| \le |R_0|+2^{ck} \cdot (\pot(\Tc_0)-\pot(\Tc))$.\label{lem:mainbalancinglemma:item4}
\end{enumerate}

Once $R$ is empty or contains only the root node, the superbranch decomposition $\Tc$ is $k$-good and we can stop.
Until then, we repeat the following process.

Let $t$ be the top node of the stack representing $R$ (in particular, with $\desc(t) \cap R = \{t\}$ and $\anc(t) \subseteq R$).
We apply \Cref{lem:testunbalancedness} to in time $2^{\OO(k)}$ test whether $t$ is $2^{2k+1}$-balanced.
If $t$ is $2^{2k+1}$-balanced, we pop $t$ from $R$ and continue to the next iteration.
If $t$ is $2^{2k+1}$-unbalanced, we apply the algorithm of \Cref{lem:balancingonestep} to transform $\Tc$ into a superbranch decomposition $\Tc' = (T',\lmap')$ via a sequence $\seq^*$ of basic rotations so that $\Tc'$ is $k$-semigood, $\pot(\Tc') \le \pot(\Tc)-1$, $\|\seq^*\| \le 2^{ck}$, and $V_{\Tc}(\seq^*) \subseteq \desc(t)$ is a connected set in $T$ that contains $t$.
We let $\seq'$ to be the concatenation of $\seq$ and $\seq^*$, $R' = (R \setminus \{t\}) \cup V_{\Tc'}(\seq^*)$, and $R_{+}' = R_{+} \cup (V_{\Tc}(\seq^*) \cap V(T_0))$.

Because $V_{\Tc}(\seq^*)$ is a connected set in $T$, we have that $V_{\Tc'}(\seq^*)$ is a connected set in $T'$.
Furthermore, because $\{t\} \subseteq V_{\Tc}(\seq^*) \subseteq \desc(t)$, we have that the parent of the shallowest node in $V_{\Tc'}(\seq^*)$ is the parent of $t$, and therefore $R'$ is a prefix of $T'$.
Furthermore, the stack representing $R$ can be transformed to represent $R'$ by first popping $t$, and then doing a depth-first search exploring $V_{\Tc'}(\seq^*)$ starting at the shallowest node of $V_{\Tc'}(\seq^*)$.
This runs in $2^{\OO(k)} \cdot |V_{\Tc'}(\seq^*)| \le 2^{\OO(k)} \cdot \|\seq^*\| \le 2^{\OO(k)}$ time.
Note that if a node of $\Tc'$ is not in $R'$, then the subtree below it in $\Tc'$ is identical to the subtree below it in $\Tc$, and therefore $R'$ contains all $2^{2k+1}$-unbalanced nodes of $\Tc'$.

It remains to prove that $\Tc'$, $R'$, $\seq'$, and $R_{+}'$ satisfy the invariants of \Cref{lem:mainbalancinglemma:item1,lem:mainbalancinglemma:item2,lem:mainbalancinglemma:item3,lem:mainbalancinglemma:item4}.
\Cref{lem:mainbalancinglemma:item1} is obvious from $\pot(\Tc') \le \pot(\Tc)-1$.
\Cref{lem:mainbalancinglemma:item2} follows from $\|\seq^*\| \le 2^{ck}$ and $\pot(\Tc') \le \pot(\Tc)-1$.
\Cref{lem:mainbalancinglemma:item4} follows from the facts that $R'_{+} = R_{+} \cup (V_{\Tc}(\seq^*) \cap V(T_0))$, $|V_{\Tc}(\seq^*)| \le \|\seq^*\| \le 2^{ck}$, and $\pot(\Tc') \le \pot(\Tc)-1$.

For \Cref{lem:mainbalancinglemma:item3}, it is clear that $R_{+}' \subseteq R_0 \cup \trace_{\Tc_0}(\seq')$ because we constructed $R_{+}'$ from $R_{+}$ by adding $V_{\Tc}(\seq^*) \cap V(T_0) \subseteq \trace_{\Tc_0}(\seq')$.
Now, consider an ancestor $a \in \anc_{\Tc_0}(x)$ of a node $x \in V_{\Tc}(\seq^*) \cap V(T_0)$.
If $a \notin V(T)$, then $a \in \trace_{\Tc_0}(\seq)$, and therefore $a \in R_{+}$.
If $a \in V(T)$ and $a \in \desc(t)$, then $a \in V_{\Tc}(\seq^*) \cap V(T_0)$.
If $a \in V(T)$, $a \in \anc(t) \setminus \{t\}$, and $t \notin V(T_0)$, then $a \in \trace_{\Tc_0}(\seq)$ and therefore $a \in R_{+}$.
If $a \in V(T)$, $a \in \anc(t) \setminus \{t\}$, and $t \in V(T_0)$, then $t \in R_0$ and therefore $a \in R_0 \subseteq R_{+}$.
Therefore $R_0 \cup \trace_{\Tc_0}(\seq') \subseteq R_{+}'$.

This concludes the proof of the correctness of the algorithm.
To prove the running time, we observe that each iteration either (1) decreases $|R|$ by $1$ without increasing $\pot(\Tc)$, or (2) decreases $\pot(\Tc)$ by at least $1$ and increases $|R|$ by $2^{\OO(k)}$.
Therefore, the number of iterations is bounded by $|R_0| + 2^{\OO(k)} \cdot (\pot(\Tc_0)-\pot(\Tc_f))$, where $\Tc_f$ is the final superbranch decomposition.
As each iteration runs in time $2^{\OO(k)}$, the total running time is $2^{\OO(k)} \cdot (|R_0| + \pot(\Tc_0)-\pot(\Tc_f))$.
\end{proof}

%% file: rotroot.tex
\section{Inserting and deleting edges}
\label{sec:insdeledge}
We then give the methods for inserting and deleting edges to the data structure of \Cref{lem:mainintlemma}.
We start by giving a subroutine for rotating leaves of the superbranch decomposition towards the root, and then use it to implement edge insertions and deletions.

\subsection{Rotating hyperedges to the root}
\label{sec:rotroot}
The main subroutine for inserting and deleting edges will be rotating hyperedges of $\su{G}$ associated with the update to the root of the superbranch decomposition $\Tc$.
In particular, if an edge is added between vertices $u$ and $v$, then the singleton hyperedges $e_u$ and $e_v$ are rotated to the root, and if an edge $uv$ is deleted, then $e_u$, $e_v$, and $e_{uv}$ are rotated to the root.

The following lemma formally captures what ``rotating to the root'' means, and the rest of this subsection is dedicated to the proof of it.
We note that in its statement we have the seemingly arbitrary constraints $k \ge 3$, $|X| \le 3$, and $|V(X)| \le 2$.
The constraints $|X| \le 3$ and $|V(X)| \le 2$ come from the aforementioned use of the rotation operation, i.e., they are satisfied when $X = \{e_u, e_v\}$ or $X = \{e_u, e_v, e_{uv}\}$.
The constraint $k \ge 3$ is then a convenient way to ensure that $|X| + 2^{|V(X)|+k} \le 2^{2k}+1$.

\begin{lemma}
\label{lem:rotrootmain}
Let $G$ be a hypergraph with $|E(G)| \ge 3$ and $e_{\bot} \in E(G)$, and let $\Tc = (T,\lmap)$ be an $e_{\bot}$-rooted superbranch decomposition of $G$.
Suppose also that $k \ge 3$ is an integer so that $\Tc$ is $k$-good and $\wl(G) \le k$.
There is an algorithm that, given $k$ and a set of hyperedges $X \subseteq E(G) \setminus \{e_{\bot}\}$ with $|X| \le 3$ and $|V(X)| \le 2$, transforms $\Tc$ into a $k$-semigood $e_\bot$-rooted superbranch decomposition $\Tc' = (T',\lmap')$ of $G$ via a sequence $\seq$ of basic rotations, so that
\begin{itemize}
\item for all $e \in X$, $\depth_{T'}(\lmap'^{-1}(e)) = 2$,
\item $\pot(\Tc') \le \pot(\Tc) + 2^{\OO(k)} \log \|G\|$, and
\item $\|\seq\|_{\Tc} \le 2^{\OO(k)} \log \|G\|$.
\end{itemize}
The running time of the algorithm is $2^{\OO(k)} \log \|G\|$, and it returns $\seq$.
\end{lemma}

We move hyperedges in $X$ upwards in the superbranch decomposition one step at a time.
To track our progress, we define the potential of a hyperedge $e \in X$ as follows.
Let $p \in V(T)$ be the parent of $\lmap^{-1}(e)$ in $T$.
Then, we let
\[\pott_{\Tc}( e) = \depth_{T}( \lmap^{-1}(e)) - \log(|\lmap[p]|).\]
Note that $\pott_{\Tc}( e)$ can be negative, in particular, its minimum possible value is $2-\log(|E(G)|-1)$, which is achieved when $\depth_{T}( \lmap^{-1}(e)) = 2$.
When $\Tc$ is $k$-good, \Cref{lem:gooddepth} implies that $\pott_{\Tc}( e)$ is bounded from above by $2^{\OO(k)} \log \|G\|$.
In particular, this holds for the initial superbranch decomposition $\Tc$.
We also denote
\[\pott_{\Tc}( X) = \sum_{e \in X} \pott_{\Tc}( e).\]

We say that a hyperedge $e \in X$ is \emph{rotatable} if 
\begin{enumerate}
\item $\depth_{T}( \lmap^{-1}(e)) \ge 3$, and
\item there is no $e' \in X$ so that $\lmap^{-1}(e')$ is a descendant of the grandparent $g$ of $\lmap^{-1}(e)$ and $\depth_{T}( \lmap^{-1}(e')) > \depth_{T}( \lmap^{-1}(e))$.
\end{enumerate}

We then give a subroutine for rotating a rotatable hyperedge in $X$ towards the root.

\begin{lemma}
\label{lem:rotatetowardsrootonestep}
There is an algorithm that, given a rotatable hyperedge $e \in X$, in time $2^{\OO(k)}$ transforms $\Tc$ into a $k$-semigood superbranch decomposition $\Tc' = (T',\lmap')$ via a sequence $\seq$ of basic rotations, so that
\begin{itemize}
\item $\pott_{\Tc'}( X) \le \pott_{\Tc}( X)-1$,
\item $\pot(\Tc') \le \pot(\Tc) + 2^{\OO(k)} \cdot (\pott_{\Tc}( X) - \pott_{\Tc'}( X))$,
\item $\|\seq\| \le 2^{\OO(k)}$, and
\item $V_{\Tc}(\seq) \subseteq \anc_T(\lmap^{-1}(X))$.
\end{itemize}
\end{lemma}
\begin{proof}
Denote the grandparent of $\lmap^{-1}(e)$ by $g$, and note that it exists and is an internal node because $\depth_{T}( \lmap^{-1}(e)) \ge 3$.
Denote also by $X' \subseteq X$ the set of hyperedges $e' \in X$ so that $\lmap^{-1}(e')$ is a descendant of $g$.
The fact that $e$ is rotatable implies that for all $e' \in X'$, $g$ is either the parent or the grandparent of $\lmap^{-1}(e')$.
Denote by $X'' \subseteq X'$ the hyperedges $e' \in X'$ so that $g$ is the grandparent of $\lmap^{-1}(e')$, and denote the set the parents of such leaves $\lmap^{-1}(e')$ by $P''$.
In particular, each $p \in P''$ is a child of $g$.

We construct a superbranch decomposition $\Tc' = (T',\lmap')$ from $\Tc$ by successively contracting each edge $pg$, where $p \in P''$, with the contraction basic rotation.
As $|P''| \le |X| \le 3$, this can be done in $2^{\OO(k)}$ time.
Let us denote the resulting node of $T'$ by $t$.
We have that $\Tc'$ is downwards well-linked, and all nodes of $T'$ except $t$ have at most $2^{2k}+1$ children, while $t$ has at most $4 \cdot 2^{2k}+1$ children.
Furthermore, $t$ is the parent of each leaf $\lmap'^{-1}(e')$ so that $e' \in X'$, and there are no $e' \in X \setminus X'$ so that $\lmap'^{-1}(e')$ is a descendant of $t$.

Before transforming $\Tc'$ further, let us prove that it satisfies the desired changes in the potential functions $\pott$ and $\pot$.
We will afterwards transform $\Tc'$ further into $\Tc''$, with that step only ``improving'' the potentials.

We first check that the $\pott_{\Tc}( X)$ potential decreases.

\begin{claim}
$\pott_{\Tc'}(e') \le \pott_{\Tc}(e')$ for all $e' \in X$ and $\pott_{\Tc'}(e) \le \pott_{\Tc}(e)-1$. In particular, $\pott_{\Tc'}( X) \le \pott_{\Tc}( X)-1$.
\end{claim}
\begin{claimproof}
We first observe that contracting an edge does not increase the depth of any leaf.
Also, contracting an edge does not decrease the quantity $\log(|\lmap[p]|)$ for any parent $p$ of a leaf.
Therefore, $\pott_{\Tc'}(e') \le \pott_{\Tc}(e')$ for all $e' \in X$.
We then observe that the contraction of $pg$, where $p$ is the parent of $\lmap^{-1}(e)$, decreased the depth of the parent of $\lmap^{-1}(e)$ by one, and therefore $\pott_{\Tc'}(e) \le \pott_{\Tc}(e)-1$.
\end{claimproof}

Then, we bound the increase in the potential $\pot$.

\begin{claim}
$\pot(\Tc') \le \pot(\Tc) + 2^{\OO(k)} \cdot (\pott_{\Tc}(X)-\pott_{\Tc'}(X))$.
\end{claim}
\begin{claimproof}
We first observe that 
\begin{align*}
\pot(\Tc') &= \pot(\Tc) + \pot_{\Tc'}( t) - \pot_{\Tc}(g) - \sum_{p \in P''} \pot_{\Tc}( p)\\
&= \pot(\Tc) + (\ddeg_{T}(g)-1+\sum_{p \in P''} (\ddeg_{T}(p)-1)) \cdot \log(|\lmap[g]|)  - \pot_{\Tc}(g) - \sum_{p \in P''} \pot_{\Tc}( p) \\
&= \pot(\Tc) + \log(|\lmap[g]|) \cdot \sum_{p \in P''} (\ddeg_{T}(p)-1) - \sum_{p \in P''} \pot_{\Tc}( p)\\
&= \pot(\Tc) + \sum_{p \in P''} (\ddeg_{T}(p)-1) \cdot (\log(|\lmap[g]|)-\log(|\lmap[p]|))\\
&\le \pot(\Tc) + 2^{\OO(k)} \cdot \sum_{p \in P''} (\log(|\lmap[g]|)-\log(|\lmap[p]|)).
\end{align*}

Then, we observe that when $p$ is the parent of $\lmap^{-1}(e')$, where $e' \in X''$,
\begin{align*}
\pott_{\Tc}(e') - \pott_{\Tc'}(e') &= \depth_{T}( p)+1-\log(|\lmap[p]|) - (\depth_{T}( p)-\log(|\lmap[g]|))\\
&= 1 + \log(|\lmap[g]|) - \log(|\lmap[p]|),
\end{align*}
implying 
\[\pott_{\Tc}( X) - \pott_{\Tc'}( X) \ge \sum_{p \in P''} (\log(|\lmap[g]|)-\log(|\lmap[p]|)),\]
finishing the claim.
\end{claimproof}

Recall that $X' \subseteq X$ is now the set of hyperedges $e' \in X$ so that $\lmap'^{-1}(e')$ is a child of $t$ in $T'$, and $X'$ contains all hyperedges $e' \in X$ so that $\lmap'^{-1}(e')$ is a descendant of $t$ in $T'$.

We then apply the operation of \Cref{lem:rotatemain} with the node $t$ and the set of children $\lmap'^{-1}(X')$.
Note that by the promise that $|V(X)| \le 2$, we have that $|\bigcup_{\ell \in \lmap'^{-1}(X')} \adh(\ell t)| \le 2$.
The operation runs in time $2^{\OO(k)}$, and transforms $\Tc'$, via a sequence $\seq^*$ of basic rotations, to a superbranch decomposition $\Tc'' = (T'', \lmap'')$ so that
\begin{enumerate}
\item $\Tc''$ is downwards well-linked,\label{lem:rotatetowardsrootonestep:enum1:item1}
\item $V_{\Tc'}(\seq^*) = \{t\}$,\label{lem:rotatetowardsrootonestep:enum1:item2}
\item for each $e' \in X'$, $\lmap^{-1}(e')$ is a child of the shallowest node in $V_{\Tc''}(\seq^*)$ in $\Tc''$, and\label{lem:rotatetowardsrootonestep:enum1:item3}
\item $\ddeg_{\Tc''}(V_{\Tc''}(\seq^*)) \le \max(3 + 2^{2+\wl(G)}, 1+2^{2 \wl(G)}) \le 2^{2k}+1$.\label{lem:rotatetowardsrootonestep:enum1:item4} (here we use that $k \ge 3$)
\end{enumerate}

The fact that all other nodes of $\Tc'$ than $t$ had at most $2^{2k}+1$ children implies with \Cref{lem:rotatetowardsrootonestep:enum1:item2,lem:rotatetowardsrootonestep:enum1:item4} that all nodes of $\Tc''$ have at most $2^{2k}+1$ children, which implies with \Cref{lem:rotatetowardsrootonestep:enum1:item1} that $\Tc''$ is $k$-semigood.

We then consider the potential functions.
\begin{claim}
\label{lem:rotatetowardsrootonestep:claim4}
$\pott_{\Tc''}( X) = \pott_{\Tc'}( X)$.
\end{claim}
\begin{claimproof}
Let $t' \in V_{\Tc''}(\seq^*)$ be the shallowest node in $V_{\Tc''}(\seq^*)$.
For all $e' \in X'$, $t'$ is the parent of $\lmap''^{-1}(e')$, and we have that $\depth_{T''}(t') = \depth_{T'}(t)$ and $\lmap''[t'] = \lmap'[t]$, implying $\pott_{\Tc''}( e') = \pott_{\Tc'}( e')$.
For all $e' \in X \setminus X'$, $t$ is not an ancestor of $\lmap'^{-1}(e')$ in $T'$, and therefore \Cref{lem:rotatetowardsrootonestep:enum1:item2} implies that $\pott_{\Tc''}(e') = \pott_{\Tc'}(e')$.
\end{claimproof}

\begin{claim}
\label{lem:rotatetowardsrootonestep:claim5}
$\pot(\Tc'') \le \pot(\Tc')$.
\end{claim}
\begin{claimproof}
This follows from the facts that (1) $\ddeg_{T'}(t)-1 = \sum_{t' \in V_{\Tc''}(\seq^*)} (\ddeg_{T''}(t')-1)$ and (2) $|\lmap'[t]| \ge |\lmap''[t']|$ for all $t' \in V_{\Tc''}(\seq^*)$.
\end{claimproof}

\Cref{lem:rotatetowardsrootonestep:claim4,lem:rotatetowardsrootonestep:claim5} complete the proof that the resulting superbranch decomposition $\Tc''$ satisfies the required properties.
The sequence $\seq$ is constructed by prepending the basic rotations contracting $pg$ for $p \in P''$ to the sequence $\seq^*$.
Clearly, $\|\seq\| \le 2^{\OO(k)} + \|\seq^*\| \le 2^{\OO(k)}$.
Furthermore, we can see that $V_{\Tc}(\seq) \subseteq \anc_T(\lmap^{-1}(X))$.
\end{proof}

We then finish the proof of \Cref{lem:rotrootmain} by giving an algorithm that repeatedly applies the operation of \Cref{lem:rotatetowardsrootonestep}.

\begin{proof}[Proof of \Cref{lem:rotrootmain}.]
We apply the algorithm of \Cref{lem:rotatetowardsrootonestep} repeatedly with a rotatable hyperedge $e \in X$, as long as there are such rotatable hyperedges.
Denote by $\Tc' = (T',\lmap')$ the resulting superbranch decomposition and by $\seq$ the sequence of basic rotations formed by concatenating the sequences outputted by \Cref{lem:rotatetowardsrootonestep}.

After there are no more rotatable hyperedges in $X$, $\depth_{T'}(\lmap'^{-1}(e)) = 2$ holds for all $e \in X$.
Furthermore, from the facts that the minimum value of $\pott_{\Tc}( X)$ is $|X|(2-\log(|E(G)|-1)) \ge -3 \log \|G\|$, the maximum value of $\pott_{\Tc}( X)$ is $2^{\OO(k)} \log \|G\|$, and the guarantees on the changes of $\pott$ and $\pot$ given by \Cref{lem:rotatetowardsrootonestep}, we deduce that there are at most $2^{\OO(k)} \log \|G\|$ iterations in this process and $\pot(\Tc') \le \pot(\Tc) + 2^{\OO(k)} \log \|G\|$.
This also implies that $\|\seq\| \le 2^{\OO(k)} \log \|G\|$.

To argue that the total running time of the process is $2^{\OO(k)} \log \|G\|$, it suffices to argue that in each iteration we can find a rotatable hyperedge in $X$, if one exists, efficiently.
For a given hyperedge $e \in X$, we can in $\OO(1)$ time check if $\depth_{T}(\lmap^{-1}(e)) \ge 3$ by using the pointers stored in the representation of $\Tc$.
We can also in $\OO(1)$ time find the quantity $|\lmap[g]|$, where $g$ is the grandparent of $\lmap^{-1}(e)$.
We observe that the hyperedge in $X$ that minimizes $|\lmap[g]|$ is rotatable, if any hyperedge in $X$ is rotatable, so we simply pick such a hyperedge.
This runs in total time $\OO(1)$ as $|X| \le 3$.

Finally, we must argue that $|\trace_{\Tc}(\seq)| \le 2^{\OO(k)} \log \|G\|$, which then also implies $\|\seq\|_{\Tc} \le 2^{\OO(k)} \log \|G\|$.
For this, we use the following fact.
If $\seq_i$ is the sequence of basic rotations corresponding to the $i$-th application of \Cref{lem:rotatetowardsrootonestep}, then $V_{\Tc_i}(\seq_i) \subseteq \anc_{T_i}(\lmap_i^{-1}(X))$, where $\Tc_i = (T_i, \lmap_i)$ is the superbranch decomposition before the $i$-th application
This implies that $\anc_{T_{i+1}}(\lmap_{i+1}^{-1}(X)) \cap V(T) \subseteq \anc_{T_i}(\lmap_i^{-1}(X)) \cap V(T)$, i.e., the ancestors of $\lmap_{i+1}^{-1}(X)$ in the updated decomposition $\Tc_{i+1}$ that are also nodes of the initial decomposition $\Tc$ are a subset of those that are ancestors of $\lmap_{i}^{-1}(X)$ in $\Tc_i$.
Therefore, $\trace_{\Tc}(\seq) \subseteq \anc_{T}(\lmap^{-1}(X))$, which by the fact that $\Tc$ is $k$-good and \Cref{lem:gooddepth} implies that $|\trace_{\Tc}(\seq)| \le 2^{\OO(k)} \log \|G\|$.
\end{proof}

%% file: insdeledge.tex
\subsection{Inserting edges}
We then give the subroutine for inserting an edge.

\begin{lemma}
\label{lem:insedge}
Let $G$ be a graph and $\Tc = (T,\lmap)$ an $e_{\bot}$-rooted superbranch decomposition of $\su{G}$.
Suppose also that $k \ge 3$ is an integer so that $\Tc$ is $k$-good and $\wl(\su{G}) \le k$.
There is an algorithm that takes as input $k$ and a new edge $uv \in \binom{V(G)}{2} \setminus E(G)$.
It assumes that $\wl(\su{G'}) \le k$, where $G'$ denotes the graph $G$ with $uv$ added.
It transforms $\Tc$ into a $k$-semigood superbranch decomposition $\Tc' = (T',\lmap')$ of $\su{G'}$ with a sequence $\seq$ of basic rotations, so that
\begin{itemize}
\item $\pot(\Tc') \le \pot(\Tc) + 2^{\OO(k)} \log \|G\|$ and
\item $\|\seq\|_{\Tc} \le 2^{\OO(k)} \log \|G\|$.
\end{itemize}
It runs in time $2^{\OO(k)} \log \|G\|$ and returns $\seq$.
\end{lemma}
\begin{proof}
We first apply the algorithm of \Cref{lem:rotrootmain} to move the hyperedges $e_u$ and $e_v$ to the root.
In particular, by applying it with the set $X = \{e_u, e_v\}$, we obtain in time $2^{\OO(k)} \log \|G\|$ a $k$-semigood $e_{\bot}$-rooted superbranch decomposition $\Tc_1 = (T_1, \lmap_1)$ via a sequence $\seq_1$ of basic rotations, so that,
\begin{itemize}
\item for both $e \in \{e_u, e_v\}$, $\depth_{T_1}(\lmap_1^{-1}(e)) = 2$,
\item $\pot(\Tc_1) \le \pot(\Tc) + 2^{\OO(k)} \log \|G\|$, and
\item $\|\seq_1\|_{\Tc} \le 2^{\OO(k)} \log \|G\|$.
\end{itemize}

Let $r \in \vint(T_1)$ be the child of the root of $T_1$, i.e., the unique node having depth $1$.
Both $\lmap_1^{-1}(e_u)$ and $\lmap_1^{-1}(e_v)$ are children of $r$.
We transform $\Tc_1$ into a superbranch decomposition $\Tc_2 = (T_2, \lmap_2)$ of $G'$ by inserting a new leaf $\lmap_2^{-1}(e_{uv})$ corresponding to $e_{uv}$ as a child of $r$ with the leaf insertion basic rotation in $2^{\OO(k)}$ time.

\begin{claim}
$\Tc_2$ is downwards well-linked.
\end{claim}
\begin{claimproof}
Let $(\lmap_2(\vec{tp}), \lmap_2(\vec{pt}))$ be a separation of $\Tc_2$, where $p$ is the parent of $t$.
If $p$ is the root, then $\bd(\lmap_2(\vec{tp})) = \emptyset$ so $\lmap_2(\vec{tp})$ is trivially well-linked.
Also, if $t$ is a leaf, then $\lmap_2(\vec{tp})$ is also trivially well-linked.
It remains to consider the case where $p$ is a descendant of $r$ and $t$ is not a leaf.
In this case, we have that $\lmap_2(\vec{tp}) = \lmap_1(\vec{tp})$.
The facts that none of $\lmap_2^{-1}(e_u)$, $\lmap_2^{-1}(e_v)$, or $\lmap_2^{-1}(e_{uv})$ are descendants of $t$ imply also that $V(\lmap_2(\vec{pt})) = V(\lmap_1(\vec{pt}))$, implying that $\bdc_{G'}(Y) = \bdc_G(Y)$ for all $Y \subseteq \lmap_2(\vec{tp})$, implying that $\lmap_2(\vec{tp})$ is well-linked in $G'$ because $\lmap_1(\vec{tp})$ is well-linked in $G$.
\end{claimproof}

We observe that all nodes of $T_2$ except $r$ have the same number of children as they had in $T_1$.
In particular, the only reason why $\Tc_2$ is not $k$-semigood is that $r$ might have more than $2^{2k}+1$ children.
We also observe that $\pot(\Tc_2) \le \pot(\Tc_1)+2^{\OO(k)} \log \|G\|$, as the insertion of the new leaf increased $|\lmap_2(r)|$ and $\ddeg_{T_2}(r)$ by one (compared to $\Tc_1$), but did not change these for any other internal nodes.

We then apply the rotation of \Cref{lem:rotatemain} with the node $r$ and an empty set of children of $r$ to transform $\Tc_2$ into a superbranch decomposition $\Tc'$ via a sequence $\seq'$ of basic rotations so that $\Tc'$ is downwards well-linked, $V_{\Tc_2}(\seq') = \{r\}$, and $\ddeg_{T'}(V_{\Tc'}(\seq')) \le 2^{2 \wl(G')} + 1 \le 2^{2k} + 1$, implying that $\Tc'$ is $k$-semigood.
The running time of this is $2^{\OO(k)}$ and we also have that $\|\seq'\|_{\Tc_2} \le 2^{\OO(k)}$.

It holds that $\pot(\Tc') \le \pot(\Tc_2) + \pot_{\Tc'}(V_{\Tc'}(\seq'))$.
By the fact that $V_{\Tc_2}(\seq') = \{r\}$, we have that $\pot_{\Tc'}(V_{\Tc'}(\seq')) \le 2^{\OO(k)} \log \|G\|$, and therefore $\pot(\Tc') \le \pot(\Tc_2) + 2^{\OO(k)} \log \|G\|$.
We construct the sequence $\seq$ by prepending $\seq_1$ and the leaf insertion basic rotation to $\seq'$.
We have that $\|\seq\|_{\Tc} \le \|\seq_1\|_{\Tc} + \|\seq'\|_{\Tc_2}+2^{\OO(k)} \le 2^{\OO(k)} \log \|G\|$.
The algorithm runs in total $2^{\OO(k)} \log \|G\|$ time.
\end{proof}

\subsection{Deleting edges}
We then give the subroutine for deleting an edge, which is similar to the subroutine for inserting an edge.

\begin{lemma}
\label{lem:deledge}
Let $G$ be graph and $\Tc = (T,\lmap)$ an $e_{\bot}$-rooted superbranch decomposition of $\su{G}$.
Suppose also that $k \ge 3$ is an integer so that $\Tc$ is $k$-good and $\wl(\su{G}) \le k$.
There is an algorithm that takes as input $k$ and an edge $uv \in E(G)$.
It assumes that $\wl(\su{G'}) \le k$, where $G'$ denotes the graph $G$ with $uv$ deleted.
It transforms $\Tc$ into a $k$-semigood superbranch decomposition $\Tc' = (T',\lmap')$ of $\su{G'}$ with a sequence $\seq$ of basic rotations, so that 
\begin{itemize}
\item $\pot(\Tc') \le \pot(\Tc) + 2^{\OO(k)} \log \|G\|$ and
\item $\|\seq\|_{\Tc} \le 2^{\OO(k)} \log \|G\|$.
\end{itemize}
It runs in time $2^{\OO(k)} \log \|G\|$ and returns $\seq$.
\end{lemma}
\begin{proof}
We first apply the algorithm of \Cref{lem:rotrootmain} to move the hyperedges $e_u$, $e_v$, and $e_{uv}$ to the root.
In particular, by applying it with the set $X = \{e_u, e_v, e_{uv}\}$, we obtain in time $2^{\OO(k)} \log \|G\|$ a $k$-semigood $e_{\bot}$-rooted superbranch decomposition $\Tc_1 = (T_1, \lmap_1)$ via a sequence $\seq_1$ of basic rotations, so that,
\begin{itemize}
\item for all $e \in \{e_u, e_v, e_{uv}\}$, $\depth_{T_1}(\lmap_1^{-1}(e)) = 2$,
\item $\pot(\Tc_1) \le \pot(\Tc) + 2^{\OO(k)} \log \|G\|$, and
\item $\|\seq_1\|_{\Tc} \le 2^{\OO(k)} \log \|G\|$.
\end{itemize}

Let $r \in \vint(T_1)$ be the child of the root of $T_1$, i.e., the unique node having depth $1$.
Note that $r$ has at least $3$ children, as each $\lmap_1^{-1}(e_u)$, $\lmap_1^{-1}(e_v)$, and $\lmap_1^{-1}(e_{uv})$ is a child of $r$.
We transform $\Tc_1$ into a superbranch decomposition $\Tc' = (T', \lmap')$ by deleting the leaf $\lmap_1^{-1}(e_{uv})$ with the leaf deletion basic rotation in $2^{\OO(k)}$ time.

\begin{claim}
$\Tc'$ is downwards well-linked.
\end{claim}
\begin{claimproof}
Let $(\lmap'(\vec{tp}), \lmap'(\vec{pt}))$ be a separation of $\Tc'$, where $p$ is the parent of $t$.
If $p$ is the root, then $\bd(\lmap'(\vec{tp})) = \emptyset$ so $\lmap'(\vec{tp})$ is trivially well-linked.
Also, if $t$ is a leaf, then $\lmap'(\vec{tp})$ is also trivially well-linked.
It remains to consider the case where $p$ is a descendant of $r$ and $t$ is not a leaf.
In this case, we have that $\lmap'(\vec{tp}) = \lmap_1(\vec{tp})$.
The facts that none of $\lmap_1^{-1}(e_u)$, $\lmap_1^{-}(e_v)$, or $\lmap_1^{-1}(e_{uv})$ are descendants of $t$ in $\Tc_1$ imply also that $V(\lmap'(\vec{pt})) = V(\lmap_1(\vec{pt}))$, implying that $\bdc_{G'}(Y) = \bdc_G(Y)$ for all $Y \subseteq \lmap'(\vec{tp})$, implying that $\lmap'(\vec{tp})$ is well-linked in $G'$ because $\lmap_1(\vec{tp})$ is well-linked in $G$.
\end{claimproof}

As the number of children of each node of $\Tc'$ is at most the number in $\Tc_1$, it follows that $\Tc'$ is $k$-semigood.
It is also easy to see that $\pot(\Tc') \le \pot(\Tc_1) \le \pot(\Tc) + 2^{\OO(k)} \log \|G\|$.

We obtain $\seq$ by appending the leaf deletion basic rotation to the sequence $\seq_1$.
We have that $\|\seq\|_{\Tc} \le \|\seq_1\|_{\Tc} + 2^{\OO(k)} \le 2^{\OO(k)} \log \|G\|$.
\end{proof}

%% file: together.tex
\section{Putting the main data structure together}
\label{sec:mainlemmaproof}
In this section we finally prove \Cref{lem:mainintlemma} by putting together the ingredients developed in \Cref{sec:balancing,sec:insdeledge}.

Let us start with a lemma providing the initialization routine with an edgeless graph.
Note that a superbranch decomposition that is $1$-good is $k$-good for all $k \ge 1$.

\begin{lemma}
\label{lem:initedgesless}
There is an algorithm that, given an edgeless graph $G$, in time $\OO(\|G\|)$ returns an $e_{\bot}$-rooted superbranch decomposition $\Tc$ of $\su{G}$ that is $1$-good and has $\pot(\Tc) \le \OO(\|G\|)$.
\end{lemma}
\begin{proof}
Recall that the set of hyperedges of $\su{G}$ consists of the singleton hyperedges $e_v$ for each $v \in V(G)$, and of the special hyperedge $e_{\bot}$.
We construct a superbranch decomposition $\Tc = (T, \lmap)$ by taking a balanced binary tree with $|V(G)|$ leaves, assigning the singleton hyperedges $e_v$ arbitrary with the leaves, and inserting a root node to which $e_{\bot}$ is assigned adjacent to the root.
A standard construction of a balanced binary tree ensures that all nodes of $\Tc$ are $3$-balanced.
Furthermore, because all hyperedges of $\su{G}$ have disjoint vertex sets, all adhesions of $\Tc$ are empty, implying that $\Tc$ is downwards well-linked.
It follows that $\Tc$ is $1$-good.

Furthermore, we have that 
\begin{align*}
\pot(\Tc) &\le \OO\left(\sum_{i=1}^{\lceil \log |V(G)| \rceil} \frac{|V(G)|}{2^i} \cdot i\right) \le \OO\left(|V(G)| \cdot \sum_{i=1}^{\infty} \frac{i}{2^i}\right) \le \OO(|V(G)|).
\end{align*}

Clearly, this construction can be implemented in $\OO(\|G\|)$ time.
\end{proof}

We then re-state and prove \Cref{lem:mainintlemma}.

\mainintlemma*
\begin{proof}
We assume without loss of generality that $k \ge 3$ (in order to apply \Cref{lem:insedge,lem:deledge}).
We will maintain a $k$-good $e_{\bot}$-rooted superbranch decomposition $\Tc$ of $\su{G}$, and analyze the amortized running time using the potential function $\pot(\Tc)$.

First, the $\mathsf{Init}(G,k)$ operation is implemented by applying \Cref{lem:initedgesless}.
This runs in $\OO(\|G\|)$ time, and results in $\Tc$ having initial potential $\pot(\Tc) \le \OO(\|G\|)$.

Then we consider the $\mathsf{AddEdge}(uv)$ operation.
We apply \Cref{lem:insedge}, which in time $2^{\OO(k)} \log \|G\|$ transforms $G$ into the graph $G'$ resulting from adding $uv$, and $\Tc$ into a $k$-semigood superbranch decomposition $\Tc_1$ of $\su{G'}$, with a sequence $\seq_1$ of basic rotations, so that
\begin{itemize}
\item $\pot(\Tc_1) \le \pot(\Tc) + 2^{\OO(k)} \log \|G\|$ and
\item $\|\seq_1\|_{\Tc} \le 2^{\OO(k)} \log \|G\|$.
\end{itemize}

It remains to turn $\Tc_1$ from $k$-semigood into $k$-good, which we do with the balancing procedure of \Cref{lem:mainbalancinglemma}.
Note that $\trace_{\Tc_1}(\seq_1)$ is a prefix of $\Tc_1$ that contains all of its $2^{2k+1}$-unbalanced nodes.
Furthermore, $|\trace_{\Tc_1}(\seq_1)| \le \|\seq_1\|_{\Tc} \le 2^{\OO(k)} \log \|G\|$, and we can compute $\trace_{\Tc_1}(\seq_1)$ from $\seq_1$ in $\OO(\|\seq_1\|_{\Tc}) = 2^{\OO(k)} \log \|G\|$ time.

We then apply the algorithm of \Cref{lem:mainbalancinglemma} with $\trace_{\Tc_1}(\seq_1)$ to transform $\Tc_1$ into a $k$-good $e_{\bot}$-rooted superbranch decomposition $\Tc_2$ of $\su{G'}$ via a sequence $\seq_2$ of basic rotations, so that
\begin{itemize}
\item $\pot(\Tc_2) \le \pot(\Tc_1) \le \pot(\Tc) + 2^{\OO(k)} \log \|G\|$ and
\item $\|\seq_2\|_{\Tc_1} \le 2^{\OO(k)} \cdot (|\trace_{\Tc_1}(\seq_1)| + \pot(\Tc_1) - \pot(\Tc_2)) \le 2^{\OO(k)} \log \|G\| + 2^{\OO(k)} \cdot (\pot(\Tc) - \pot(\Tc_2))$.
\end{itemize}

The running time of this is 
\begin{align*}
&2^{\OO(k)} \cdot (|\trace_{\Tc_1}(\seq_1)| + \pot(\Tc_1) - \pot(\Tc_2))\\
\le &2^{\OO(k)} \log \|G\| + 2^{\OO(k)} \cdot (\pot(\Tc) - \pot(\Tc_2)).
\end{align*}

We then append $\seq_2$ to $\seq_1$ to obtain a sequence $\seq$ of basic rotations that transforms $\Tc$ into $\Tc_2$.
We have that
\begin{align*}
\|\seq\|_{\Tc} &\le \|\seq_1\|_{\Tc} + \|\seq_2\|_{\Tc_1}\\
&\le 2^{\OO(k)} \log \|G\| + 2^{\OO(k)} \cdot (\pot(\Tc) - \pot(\Tc_2)).
\end{align*}

We then return $\seq$.
This concludes the description of the $\mathsf{AddEdge}(uv)$ operation.
The running time of the operation is $2^{\OO(k)} \log \|G\| + 2^{\OO(k)} \cdot (\pot(\Tc) - \pot(\Tc_2))$, and it increases the potential by at most $2^{\OO(k)} \log \|G\|$, i.e., we have $\pot(\Tc_2) \le \pot(\Tc) + 2^{\OO(k)} \log \|G\|$.


The $\mathsf{DeleteEdge}(uv)$ operation is implemented in the exact same way as the $\mathsf{AddEdge}(uv)$ operation, except using \Cref{lem:deledge} instead of \Cref{lem:insedge}.
In particular, it also has running time $2^{\OO(k)} \log \|G\| + 2^{\OO(k)} \cdot (\pot(\Tc) - \pot(\Tc_2))$, and increases the potential by at most $2^{\OO(k)} \log \|G\|$.

From the aforementioned running times and properties of the potential function $\pot(\Tc)$ it follows that both $\mathsf{AddEdge}(uv)$ and $\mathsf{DeleteEdge}(uv)$ have amortized running time $2^{\OO(k)} \log \|G\|$, and $\mathsf{Init}(G,k)$ has amortized running time $2^{\OO(k)} \|G\|$.
\end{proof}

%% file: prds.tex
\section{From superbranch decompositions to tree decompositions}
\label{sec:prds}
In \Cref{sec:struct,sec:balancing,sec:insdeledge,sec:mainlemmaproof} we gave the main technical contribution of this paper, namely, the proof of \Cref{lem:mainintlemma}.
In this section, we provide wrappers around \Cref{lem:mainintlemma} to finish the proof of \Cref{the:maintheorem}.
First, in \Cref{subsec:prds} we present our framework for formalizing the maintenance of dynamic programming schemes on the tree decomposition, which is based on~\cite{DBLP:conf/focs/KorhonenMNP023}, and then in \Cref{subsec:bwtotw} we translate the setting of superbranch decompositions to the setting of tree decompositions.

\subsection{Manipulating dynamic tree decompositions}
\label{subsec:prds}
We review the definitions of \emph{annotated tree decompositions}, \emph{prefix-rebuilding updates}, \emph{prefix-rebuilding data structures}, and \emph{tree decomposition automata}, which were introduced in~\cite{DBLP:conf/focs/KorhonenMNP023}.
Our definitions are not completely identical to the ones given in~\cite{DBLP:conf/focs/KorhonenMNP023}, but the results from therein still easily translate to our setting.

\paragraph{Annotated tree decompositions.}
We will manipulate \emph{annotated tree decompositions} of graphs.
An annotated tree decomposition of a graph $G$ is a triple $(T,\bag,\edges)$, so that
\begin{itemize}
\item $T$ is a binary tree, i.e., $T$ is rooted and $\ddeg(T) \le 2$,
\item $\bag \colon V(T) \to 2^{V(G)}$ is a function so that $(T,\bag)$ is a tree decomposition of $G$, and
\item $\edges \colon V(T) \to 2^{E(G)}$ is a function so that for all $t \in V(T)$, the set $\edges(t)$ contains the edges $uv$ of $G$ for which $t$ is the unique smallest-depth node with $u,v \in \bag(t)$.
\end{itemize}

Note that $G$ and $(T,\bag)$ define the annotated tree decomposition $(T,\bag,\edges)$ uniquely.
Also, $(T,\bag,\edges)$ defines $G$ uniquely.
Because $(T,\bag,\edges)$ defines $G$, updates to an annotated tree decomposition also encode updates to the graph $G$.
For a set $X \subseteq V(T)$, the \emph{restriction} of $(T,\bag,\edges)$ to $X$ is the tuple $\funrestriction{(T,\bag,\edges)}{X} = (T[X], \funrestriction{\bag}{X}, \funrestriction{\edges}{X})$.

\paragraph{Prefix-rebuilding updates.}
An update that changes an annotated tree decomposition $\Tc = (T,\bag,\edges)$ into an annotated tree decomposition $\Tc' = (T', \bag', \edges')$ is a \emph{prefix-rebuilding update} with prefixes $P$ and $P'$ if
\begin{itemize}
\item $P \subseteq V(T)$ is a prefix of $T$,
\item $P' \subseteq V(T')$ is a prefix of $T'$, and
\item $\funrestriction{(T,\bag,\edges)}{V(T)\setminus P} = \funrestriction{(T',\bag',\edges')}{V(T') \setminus P'}$.
\end{itemize}

In particular, a prefix-rebuilding update replaces the prefix $P$ by a new prefix $P'$.
A \emph{description} of a prefix-rebuilding update is a triple $\prud = (P, \Tc^{\star}, \pi)$, where $P$ is the prefix of $T$ as mentioned above, $\Tc^{\star} = (T^{\star}, \bag^{\star}, \edges^{\star})$ is an annotated tree decomposition so that
\begin{itemize}
\item $\Tc^{\star} = \funrestriction{\Tc'}{P'}$,
\end{itemize}
i.e., $\Tc^{\star}$ describes the new annotated tree decomposition for the nodes in $P'$, in particular, $V(T^{\star}) = P'$, and
\begin{itemize}
\item $\pi$ is a function that maps each node $t \in V(T) \setminus P$ that has a parent in $P$ into a node $\pi(t) \in V(T^{\star})$ that is the parent of $t$ in $T'$.
\end{itemize}

We observe that $\Tc'$ can be uniquely determined given $\Tc$ and $\prud$.
We define the \emph{size} of $\prud$ as $|\prud| = |P| + |P'|$.
We also observe that when both $\Tc$ and $\Tc'$ have width at most $k$, a representation of $\Tc$ can be turned into a representation of $\Tc'$ in time $k^{\OO(1)} \cdot |\prud|$.

\paragraph{Prefix-rebuilding data structures.}
A \emph{prefix-rebuilding data structure with overhead $\overheadtime$} is a dynamic data structure that stores an annotated tree decomposition $\Tc$, and supports at least the following two operations:

\begin{itemize}
\item $\mathsf{Init}(\Tc)$: Initializes the data structure with a given annotated tree decomposition $\Tc$. Runs in time $\overheadtime(\width(\Tc)) \cdot \|\Tc\|$.
\item $\mathsf{Update}(\prud)$: Updates the stored annotated tree decomposition $\Tc$ into a new annotated tree decomposition $\Tc'$ with a prefix-rebuilding update described by $\prud$. Runs in time $\overheadtime(\max(\width(\Tc), \width(\Tc'))) \cdot |\prud|$.
\end{itemize}

To be useful, a prefix-rebuilding data structure should also support some additional operations.
In the work of~\cite{DBLP:conf/focs/KorhonenMNP023}, prefix-rebuilding data structures were used for several internal aspects of their data structure.
In this work, we use them only for the application of maintaining tree decomposition automata.

\paragraph{Tree decomposition automata.}
There is a long history of formalizing dynamic programming on tree decompositions through automata, see for example~\cite{DBLP:journals/iandc/Courcelle90,Courcelle:2012book}, \cite[Chapter~12]{DowneyF13}, and~\cite[Chapters~10~and~11]{FlumGrohebook}.
In this paper, we use the definition of tree decomposition automaton given in~\cite{DBLP:conf/focs/KorhonenMNP023}.
Due to the length of this definition, we present it formally only in \Cref{sec:apptdautom}, but let us here give an informal definition.

A tree decomposition automaton $\autom$ processes annotated tree decompositions in a manner so that the state of $\autom$ on a node $t$, denoted by $\run_{\autom}(t)$, depends only on $\bag(t)$, $\edges(t)$, the states of $\autom$ on the child nodes of $t$, and the bags of the child nodes of $t$.
If computing the state of $\autom$ on a node $t$ on a tree decomposition of width $\le k$, based on this information, takes time at most $\evaltime(k)$, and furthermore the state can be represented in space $\evaltime(k)$, then we say that the \emph{evaluation time} of $\autom$ is $\evaltime$.
A \emph{run} of $\autom$ of an annotated tree decomposition $(T,\bag,\edges)$ is the mapping $\run_{\autom} \colon V(T) \to Q$, where $Q$ is the state set of $\autom$.

We do not assume that the state set $Q$ is small, only that each state can be represented in space $\evaltime(k)$ on tree decompositions of width $k$.
Note that therefore, a representation of a state may consist of $\Omega(\evaltime(k) \cdot \log n)$ bits.
As for the representation of $\autom$, we assume that it is given as a word RAM machine that implements the state transitions.

With these definitions, we can state the following lemma from~\cite{DBLP:conf/focs/KorhonenMNP023}, that states that runs of tree decomposition automata can be efficiently maintained under prefix-rebuilding updates.

\begin{lemma}[{\cite[Lemma~A.6]{DBLP:conf/focs/KorhonenMNP023}}]
\label{lem:prdsautomaton}
Given a tree decomposition automaton $\autom$ with evaluation time $\evaltime$, we can construct a prefix-rebuilding data structure with overhead $\overheadtime(k) = \evaltime(k) \cdot k^{\OO(1)}$, that in addition to the $\mathsf{Init}$ and $\mathsf{Update}$ operations implements the following operation:
\begin{itemize}
\item $\mathsf{Query}(t)$: Given a node $t$, return $\run_{\autom}(t)$. Runs in time $\OO(\evaltime(\width(\Tc)))$, where $\Tc$ is the current annotated tree decomposition.
\end{itemize}
\end{lemma}

%% file: bwtotw.tex
\subsection{Proof of \Cref{the:maintheorem}}
\label{subsec:bwtotw}
We then provide a wrapper around the data structure of \Cref{lem:mainintlemma} to lift it from maintaining a downwards well-linked superbranch decomposition to maintaining an annotated tree decomposition.
The main idea is to compute for each node $t$ a tree decomposition of $\primal(\torso(t))$, and then stitch them together.
The following lemma provides the subroutine for computing a tree decomposition of $\primal(\torso(t))$.
Its proof is relegated to \Cref{sec:app:missingproofs} because it follows well-known techniques~\cite{RobertsonS-GMXIII}.
Note that the torso of any node of a superbranch decomposition is always a normal hypergraph.

\begin{restatable}[\inapp]{lemma}{computetdecomp}
\label{lem:computetorsotdecomp}
There is an algorithm that, given a normal hypergraph $G$ and a hyperedge $e \in E(G)$, in time $2^{\OO(\bdc(e) + \wl_e(G))} \cdot \|G\|^{\OO(1)}$ returns tree decomposition $\Tc = (T,\bag)$ of $\primal(G)$, and an injective mapping $q \colon E(G) \to \leaves(T)$ so that
\begin{itemize}
\item $\width(\Tc) \le 3 \cdot \max(\bdc(e), \wl_e(G)) - 1$,
\item $\|\Tc\| \le \|G\|^{\OO(1)}$,
\item the maximum degree of $T$ is $3$, and
\item for all $e \in E(G)$, $V(e) \subseteq \bag(q(e))$.
\end{itemize}
\end{restatable}

Now we are ready to give our data structure in terms of treewidth.

\begin{lemma}
\label{lem:finaldyntwprefixrebuild}
Let $G$ be a dynamic graph and $k \ge 1$ an integer with a promise that $\tw(G) \le k$ at all times.
There is a data structure that maintains an annotated tree decomposition $\Tc$ of $G$ of width $\le 9 \cdot \tw(G) + 8$ and depth $\le 2^{\OO(k)} \log \|G\|$, and supports the following operations:

\begin{itemize}
\item $\mathsf{Init}(G,k)$: Given an edgeless graph $G$ and an integer $k \ge 1$, initialize the data structure with $G$ and $k$, and return $\Tc$. Runs in $2^{\OO(k)} \|G\|$ amortized time.
\item $\mathsf{AddEdge}(uv)$: Given a new edge $uv \in \binom{V(G)}{2} \setminus E(G)$, add $uv$ to $G$. Runs in $2^{\OO(k)} \log \|G\|$ amortized time.
\item $\mathsf{DeleteEdge}(uv)$: Given an edge $uv \in E(G)$, delete $uv$ from $G$. Runs in $2^{\OO(k)} \log \|G\|$ amortized time.
\end{itemize}

Furthermore, in the operations $\mathsf{AddEdge}$ and $\mathsf{DeleteEdge}$, $\Tc$ is updated by a prefix-rebuilding update, and a description $\prud$ of that update is returned.
The sizes $|\prud|$ of the descriptions have the same amortized upper bound as the running time.
\end{lemma}
\begin{proof}
We use the data structure of \Cref{lem:mainintlemma}.
We relay all of the operations to it, and let it maintain a downwards well-linked $e_{\bot}$-rooted superbranch decomposition $\tilde{\Tc} = (\tilde{T}, \tilde{\lmap})$ of $\su{G}$, so that $\ddeg(\tilde{T}) \le 2^{\OO(k)}$ and $\depth(\tilde{T}) \le 2^{\OO(k)} \log \|G\|$.
Because $\wl(\su{G}) \le 3 \cdot (\tw(G) + 1)$ (by \Cref{lem:wltotwlink}), we can set the value of $k$ in \Cref{lem:mainintlemma} to be $3k+3$.

In order to maintain the $\edges$ function of an annotated tree decomposition, we maintain a function $\el \colon V(\tilde{T}) \to 2^{E(G)}$ on $\tilde{\Tc}$.
This stores, for each internal node $t \in \vint(\tilde{T})$, all edges $uv$ of $G$ so that (1)~$e_{uv} \in \tilde{\lmap}[t]$ and (2)~$u,v \in V(\torso(t))$.
For each leaf-node $\ell$, we define $\el(\ell) = \{uv\}$ if $\tilde{\lmap}(\ell) = e_{uv}$, and $\el(\ell) = \emptyset$ otherwise.
Because $\|\torso(t)\| \le 2^{\OO(k)}$, we have that $|\el(t)| \le 2^{\OO(k)}$, and furthermore, given the values $\el(c_i)$ for all children $c_i$ of $t$ and the hypergraph $\torso(t)$, we can compute $\el(t)$ in time $2^{\OO(k)}$.

When $\tilde{\Tc}$ is updated by a sequence $\seq$ of basic rotations into a superbranch decomposition $\tilde{\Tc}'$, the function $\el$ can be recomputed by recomputing it bottom-up for all nodes in the prefix $\trace_{\tilde{\Tc}'}(\seq)$, in time $2^{\OO(k)} \cdot |\trace_{\tilde{\Tc}'}(\seq)| = 2^{\OO(k)} \cdot \|\seq\|_{\tilde{\Tc}}$.
In particular, the running time guarantees of \Cref{lem:mainintlemma} hold even while maintaining $\el$.

We maintain an annotated tree decomposition $\Tc = (T,\bag,\edges)$, that is obtained from $\tilde{\Tc}$ and $\el$ as follows.
For each internal node $t \in \vint(\tilde{T})$ with parent $p$, let $\Tc_t = (T_t, \bag_t)$ be the tree decomposition of $\primal(\torso(t))$ outputted by applying the algorithm of \Cref{lem:computetorsotdecomp} with $\torso(t)$ and the hyperedge $e_p \in E(\torso(t))$.
Let also $q_t$ be the mapping $q_t \colon E(\torso(t)) \to \leaves(T_t)$ outputted by it.
The tree $T$ is constructed as follows.
First, the nodes of $T$ are
\[V(T) = \{v_\ell \mid \ell \in \leaves(\tilde{T})\} \cup \{v_{st} \mid st \in E(\tilde{T})\} \cup \bigcup_{t \in \vint(\tilde{T})} V(T_t).\]
The edges of $T$ consist of the union of
\begin{itemize}
\item edges $v_{\ell} v_{\ell t}$, where $\ell \in \leaves(\tilde{T})$ and $\ell t \in E(\tilde{T})$ is the edge of $\tilde{T}$ incident to $\ell$,
\item the edges of $T_t$ for all $t \in \vint(\tilde{T})$, and
\item for each $t \in \vint(\tilde{T})$ and an incident edge $ts \in E(\tilde{T})$, the edge $v_{ts} q_t(e_s)$, where $e_s \in E(\torso(t))$ is the hyperedge corresponding to the edge $ts$ of $\tilde{T}$. Note that $\adh(st) = V(e_s) \subseteq \bag_t(q_t(e_s))$.
\end{itemize}

Then, the bags of $\Tc$ are constructed as
\begin{itemize}
\item for each $\ell \in \leaves(\tilde{T})$, $\bag(v_\ell) = V(\tilde{\lmap}(\ell))$,
\item for each $st \in E(\tilde{T})$, $\bag(v_{st}) = \adh(st)$, and
\item for each $t \in \vint(\tilde{T})$ and $t' \in V(T_t)$, $\bag(t') = \bag_t(t')$.
\end{itemize}

The $\edges$ function is constructed by letting $\edges(v_\ell) = \emptyset$ for all $\ell \in \leaves(\tilde{T})$, $\edges(v_{st}) = \emptyset$ for all $st \in E(\tilde{T})$, and for each $t \in \vint(\tilde{T})$ and $s \in V(T_t)$, assigning $\edges(s)$ to contain all edges $uv \in \el(t) \setminus \binom{\adh(tp)}{2}$, for which $s$ is the smallest depth node of $\Tc_t$ with $u,v \in \bag_t(s)$, when interpreting $\Tc_t$ as rooted on $q_t(e_p)$, where $p$ is the parent of $t$.

We let the root of $T$ be the node $v_{\ell}$ corresponding to the leaf $\ell$ with $\lmap(\ell) = e_{\bot}$.
It is not difficult to show that $(T,\bag)$ is indeed a tree decomposition of $G$ by applying \Cref{obs:superbdtotd}.
Furthermore, because each torso of $\tilde{\Tc}$ has size at most $2^{\OO(k)}$, we have that $\depth(T) \le 2^{\OO(k)} \cdot \depth(\tilde{T}) \le 2^{\OO(k)} \log \|G\|$.
Because each tree $T_t$ has maximum degree $3$, it follows that $T$ has maximum degree $3$, and as its root has degree $1$, it follows that $T$ is binary.
We also observe that the $\edges$ function is correctly constructed, so that $(T,\bag,\edges)$ is an annotated tree decomposition of $G$.

We then argue that the width of $\Tc$ is at most $9 \cdot \tw(G) + 8$.
By \Cref{lem:wltotwlink}, we have $\wl(\su{G}) \le 3 \cdot \tw(G) + 3$.
Because $\tilde{\Tc}$ is downwards well-linked, by \Cref{lem:wltransindecomp} we get that for each $t \in \vint(\tilde{T})$, we have $\wl_{e_p}(\torso(t)) \le \wl(\su{G})$ and $\bdc_{\torso(t)}(e_p) \le \wl(\su{G})$.
Therefore, the tree decompositions outputted by \Cref{lem:computetorsotdecomp} have width at most $3 \cdot \wl(\su{G}) - 1 \le 9 \cdot \tw(G) + 8$.

We observe that $\Tc$ can be maintained locally, in the sense that if $\tilde{\Tc}$ is transformed to $\tilde{\Tc'}$ by a sequence of basic rotations $\seq$, then the only parts of $\Tc$ that need to be recomputed are the nodes corresponding to the nodes in $\trace_{\tilde{\Tc'}}(\seq)$ and the edges incident to them.
As the torsos of $\tilde{\Tc}$ have size at most $2^{\OO(k)}$, and the algorithm of \Cref{lem:computetorsotdecomp} runs in time $2^{\OO(\bdc(e_p) + \wl_{e_p}(\torso(t)))} \cdot \|\torso(t)\|^{\OO(1)} = 2^{\OO(k)}$, this means that $\Tc$ can be updated in time $2^{\OO(k)} \cdot \|\seq\|_{\tilde{\Tc}}$ whenever $\tilde{\Tc}$ is updated by a sequence $\seq$ of basic rotations.
Furthermore, this update of $\Tc$ can be expressed as a prefix-rebuilding update with a description of size $2^{\OO(k)} \cdot \|\seq\|_{\tilde{\Tc}}$.

By the guarantees given by \Cref{lem:mainintlemma}, the values $2^{\OO(k)} \cdot \|\seq\|_{\tilde{\Tc}}$ over all the updates have the amortized upper bound of $2^{\OO(k)} \log \|G\|$ per update.
Therefore, this is also an amortized upper bound for the running time of this data structure and the sizes of the descriptions of the prefix-rebuilding updates used for maintaining it.
\end{proof}

We then combine \Cref{lem:prdsautomaton,lem:finaldyntwprefixrebuild} to conclude the proof of \Cref{the:maintheorem}.

\maintheorem*
\begin{proof}
The data structure of \Cref{lem:finaldyntwprefixrebuild} already gives all parts of this theorem except for maintaining a run of the automaton $\autom$.
For this, we use the prefix-rebuilding data structure of \Cref{lem:prdsautomaton}.
We use it so, that the data structure of \Cref{lem:prdsautomaton} always contains a copy of the tree decomposition $\Tc$ maintained by the data structure of \Cref{lem:finaldyntwprefixrebuild}.
In particular, after the initialization of the data structure of \Cref{lem:finaldyntwprefixrebuild}, we initialize the data structure of \Cref{lem:prdsautomaton} with the tree decomposition returned by \Cref{lem:finaldyntwprefixrebuild}.
Then, on each update, we pass the description of a prefix-rebuilding update returned by \Cref{lem:finaldyntwprefixrebuild} to update the tree decomposition stored by the data structure of \Cref{lem:prdsautomaton}.
Now, at all points the $\mathsf{Query}$ operation of can be used to query the states of $\autom$ on the current tree decomposition $\Tc$.
This causes an additional running time overhead of a factor of $\evaltime(9k+8) \cdot k^{\OO(1)}$.
\end{proof}

Let us also describe briefly how the corollaries mentioned in \Cref{sec:intro} are obtained.

\cordp*
\begin{proof}
For the first bullet point, we observe that the classical dynamic programming algorithms for computing the size of a maximum independent set, the size of a minimum dominating set, and $q$-colorability for constant $q$, in time $2^{\OO(k)} n$, (see e.g.~\cite[Chapter~7]{DBLP:books/sp/CyganFKLMPPS15}) can be interpreted as tree decomposition automata with evaluation time $2^{\OO(k)}$.
For the second bullet point, it follows from the work of Courcelle~\cite{DBLP:journals/iandc/Courcelle90,Courcelle:2012book} that for every graph property expressible in the counting monadic second-order logic, there is a tree decomposition automaton with evaluation time $\OO_k(1)$ that maintains whether $G$ satisfies the property.
See~\cite[Lemma~A.2]{DBLP:conf/focs/KorhonenMNP023} for a formalization of this in our framework.
\end{proof}

\corkpathfvs*
\begin{proof}
The dynamic programming algorithms for feedback vertex set and longest path by Bodlaender et al.~\cite{BodlaenderCK12} can be interpreted as tree decomposition automata for those problems with evaluation time $2^{\OO(k)}$.
We use the facts that the treewidth of a graph is at most the size of minimum feedback vertex set plus one (follows from the elementary observation that adding one vertex can increase treewidth by at most one), and that the treewidth of a graph is at most the length of its longest path~\cite{DBLP:conf/wads/Bodlaender89}.
We use the data structure of \Cref{the:maintheorem} with these two automata and the treewidth bound $9 \cdot (k+1) + 9$.

This works under the promise that the treewidth of the graph stays at most $9 \cdot (k+1) + 9$ at all times, but we have no such promise.
However, we can obtain this via applying the well-known ``delaying invariant-breaking updates'' technique~\cite{EppsteinGIS96}.
In our context, this works as follows.
If there is an edge insertion that increases the width of the maintained tree decomposition $\Tc$ to $> 9 \cdot (k+1) + 8$, then we immediately reverse it by an edge deletion, and instead move the edge to a queue $Q$ holding edges that need to be inserted.
Then, at every subsequent update, if the queue $Q$ is non-empty, we attempt to insert edges from it to the data structure, until an insertion is ``rejected'', i.e., it would increase the width of $\Tc$ to $> 9 \cdot (k+1) + 8$.
Furthermore, if the width of $\Tc$ is already $> 9 \cdot (k+1) + 8$, we do not even attempt the insertion, but directly insert the edge to $Q$.
For edge deletions, if the edge is in $Q$, it is removed from $Q$, and if it is in the data structure, it is removed from it.
This ensures that we insert edges to the data structure only when it contains a tree decomposition of width $\le 9 \cdot (k+1) + 8$, which ensures the promise that treewidth never increases to $> 9 \cdot (k+1) + 9$.
Furthermore, this still keeps the amortized running times of the updates $2^{\OO(k)} \log n$

Now, whenever $Q$ is non-empty, we have that $\tw(G) > k+1$ and therefore the graph contains a path of length $\ge k$ and does not have a feedback vertex set of size $\le k$.
Whenever $Q$ is empty, the data structure holds the entire graph, and the automata maintain the required information.
\end{proof}

\corsubexp*
\begin{proof}
There exists an integer $h_k \le \OO(\sqrt{k})$ so that every planar graph with treewidth $> h_k$ has no dominating set of size $\le k$ and contains a path of length $\ge k$~\cite{DemaineFHT05jacm}.
As recalled in the proof of \Cref{cor:twdp}, there is a tree decomposition automaton for computing the size of a minimum dominating set with evaluation time $2^{\OO(k)}$.
Furthermore, algorithm of~\cite{BodlaenderCK12} can be interpreted as a tree decomposition automaton for computing the length of a longest path with evaluation time $2^{\OO(k)}$.
We use the data structure of \Cref{the:maintheorem} with these two automata and treewidth bound $9 \cdot h_k + 9$.

This works under the promise that the treewidth of the planar graph $G$ stays at most $9 \cdot h_k + 9$ at all times, but we have no such promise.
However, we can apply the exact same ``delaying invariant-breaking updates'' technique as in the proof of \Cref{cor:corkpathfvs} to obtain this promise.
\end{proof}

%% file: conclusion.tex
\section{Conclusions}
\label{sec:conclusion}
We have given a dynamic treewidth data structure with logarithmic amortized update time for graphs of bounded treewidth.
We discuss here extensions of our result, its applications, and future directions.

\paragraph{Extensions.}
First, we note that the initialization procedure of our data structure assumes that the initial graph is edgeless.
Of course, an initialization operation with $2^{\OO(k)} n \log n$ amortized time for any $n$-vertex graph of treewidth $k$ can be obtained via inserting edges one by one, but perhaps sometimes it could be useful to initialize in $2^{\OO(k)} n$ amortized time with a given graph.
We believe that this can be done via constant-approximating treewidth~\cite{Korhonen21}, turning tree decompositions into logarithmic depth~\cite{DBLP:journals/siamcomp/BodlaenderH98}, and turning superbranch decompositions into downwards well-linked~\cite{DBLP:journals/corr/abs-2411-02658}.
However, this could get quite technical.

Second, the graph $G$ could be decorated with various labels that could be taken into account by the tree decomposition automata.
For example, \Cref{the:maintheorem} extends to maintaining the weight of a maximum independent set on vertex-weighted graphs or supporting shortest path queries on directed graphs with $2^{\OO(k)} \log n$ amortized update time.

We also recall that by applying the well-known ``delaying invariant-breaking updates'' technique~\cite{EppsteinGIS96}, as in the proof of \Cref{cor:corsubexp} (see~\cite{DBLP:conf/focs/KorhonenMNP023} for its previous application to dynamic treewidth), the data structure of \Cref{the:maintheorem} can be made resilient to the treewidth of $G$ increasing to more than $k$, in this case holding a marker \ttl instead of any other information while the treewidth of $G$ is larger than $k$.

Another direction would be to not have a pre-set treewidth bound $k$ at all, but instead let the running time of the data structure depend on the current treewidth $\tw(G)$.
We believe that with minor modifications, our data structure already achieves something along these lines, but phrasing it formally would get technical because of the amortization.

\paragraph{(Potential) applications.}
Perhaps the most significant application of the dynamic treewidth data structure of~\cite{DBLP:conf/focs/KorhonenMNP023} has been the parameterized almost-linear time algorithm for $H$-minor containment and $k$-disjoint paths by Korhonen, Pilipczuk, and Stamoulis~\cite{DBLP:conf/focs/KorhonenPS24}.
Together with the authors, we believe that by applying the logarithmic-time dynamic treewidth of \Cref{the:maintheorem} and the algorithm of~\cite{DBLP:journals/corr/abs-2411-02658}, the running time of the algorithm of~\cite{DBLP:conf/focs/KorhonenPS24} can be improved from almost-linear $\OO_{k}(m^{1+o(1)})$ to near-linear $\OO_{k}(m \polylog m)$.
The details of this remain to be written down in future work.

Another, less direct, application of the dynamic treewidth data structure of~\cite{DBLP:conf/focs/KorhonenMNP023} was its adaptation to \emph{dynamic rankwidth} by Korhonen and Sokołowski~\cite{DBLP:conf/stoc/Korhonen024}, which resulted in an $\OO_k(n^{1+o(1)}) + \OO(m)$ time algorithm for computing rankwidth, improving upon previous $\OO_k(n^2)$ time~\cite{DBLP:conf/stoc/FominK22}.
We believe that further improvements could be possible by extending the techniques of this paper to the setting of rankwidth.

Another application of dynamic treewidth in the literature is the adaptation of the Baker's technique~\cite{Baker94} for approximation schemes on planar graphs to the dynamic setting by Korhonen, Nadara, Pilipczuk, and Sokołowski~\cite{DBLP:conf/soda/KorhonenNPS24}.
They did not use a generic dynamic treewidth data structure, but a problem-specific method of using treewidth in the dynamic setting.
They achieved an update time $\OO_{\varepsilon}(n^{o(1)})$, so it would be interesting if our dynamic treewidth data structure could be used to improve this to $\OO_{\varepsilon}(\log n)$.

Potential future applications of dynamic treewidth include obtaining dynamic versions and improving the running times of the known applications of treewidth.
In addition to the ones already mentioned, this includes topics such as model checking for first-order logic~\cite{FrickG01-dec}, kernelization~\cite{BodlaenderFLPST16}, and various applications of the irrelevant vertex technique~\cite{DBLP:conf/stoc/GroheKMW11,DBLP:conf/soda/SauST25}.
Recently, the present author with coauthors~\cite{DBLP:conf/stoc/BertramHJK26} applied the techniques of this paper to provide a dynamic version of the meta-kernelization framework of Bodlaender et al.~\cite{BodlaenderFLPST16}.
Even more interesting would be applications of dynamic treewidth to settings where treewidth has not been applied before.

\paragraph{Future directions.}
The update time $2^{\OO(k)} \log n$ of our algorithm is optimal in the following sense: Dynamic forests require $\Omega(\log n)$ update time~\cite{DBLP:journals/siamcomp/PatrascuD06}, and no constant-approximation algorithms for treewidth running in time $2^{o(k)} n^{\OO(1)}$, or even in time $2^{o(n)}$, are known.
However, we could still ask if the running time could be improved to $2^{\OO(k)} + \OO(\log n)$, or to $f(k) + k^{\OO(1)} \log n$ for some function $f$.
Another natural question is whether our data structure can be de-amortized.
As the dynamic treewidth data structure of~\cite{DBLP:conf/focs/KorhonenMNP023} is also amortized, currently it is not known whether even an $\OO_k(n^{o(1)})$ worst-case update time can be achieved for maintaining tree decompositions with approximation ratio a function of $k$.

Another direction is about improving the approximation ratio.
The current ratio of $9$ comes from the factors of $3$ in both \Cref{lem:wltotwlink} and \Cref{lem:computetorsotdecomp}.
We believe that it can be shown that explicitly maintaining a tree decomposition with approximation ratio less than $3$ is not possible in $\OO_k(\log n)$ (amortized) update time, simply due to the tree decomposition requiring to change too much, and plan to write down this argument in the future.
In this light, an interesting goal would be to attain the ratio of $3$ within $\OO_k(\log n)$ amortized update time.

%% file: gtapp.tex
\section{Missing proofs}
\label{sec:app:missingproofs}
We now give proofs of some lemmas that were previously omitted due to them being standard.

\paragraph{Algorithm for well-linkedness.}
We give an algorithm for testing if a set is well-linked.
For this and the following lemma, we use the following definition of a \emph{separation} of a graph.
A separation of a graph $G$ is a pair $(A,B)$ with $A,B \subseteq V(G)$, so that $A \cup B = V(G)$ and there are no edges between $A \setminus B$ and $B \setminus A$.
The \emph{order} of a separation $(A,B)$ is $|A \cap B|$.

\lemtestwl*
\begin{proof}
For a bipartition $(C_1, C_2)$ of $A$, we call the pair $(\bd(A) \cap \bd(C_1), \bd(A) \cap \bd(C_2))$ the \emph{signature} of $(C_1, C_2)$.
There are $2^{\OO(\bdc(A))}$ different signatures, so it suffices to design a polynomial-time algorithm that, given a signature $(S_1,S_2)$, either concludes that there are no bipartitions $(C_1,C_2)$ of $A$ with $\bdc(C_i) < \bdc(A)$ with signature $(S_1,S_2)$, or returns a bipartition $(C_1,C_2)$ of $A$ with $\bdc(C_i) < \bdc(A)$ for both $i \in [2]$, with any signature.

Let $G_A$ denote the hypergraph induced by the set $A$, i.e., having $V(G_A) = V(A)$ and $E(G_A) = A$, and denote $G' = \primal(G_A)$.
We observe that if there is a bipartition $(C_1,C_2)$ of $A$ with $\bdc(C_i) < \bdc(A)$ with signature $(S_1,S_2)$, then there is a separation $(X,Y)$ of $G'$ so that $S_1 \subseteq X$, $S_2 \subseteq Y$, and $|X \cap Y| + \max(|S_1 \setminus S_2|, |S_2 \setminus S_1|) < \bdc(A)$.
Furthermore, such a separation can be found in polynomial-time, using for example the Ford-Fulkerson maximum flow algorithm.
Also, such a separation can be turned into a desired bipartition $(C_1,C_2)$ by assigning every hyperedge $e$ with $V(e) \subseteq X$ into $C_1$, and other hyperedges, for which it holds that $V(e) \subseteq Y$, into $C_2$.
\end{proof}

\paragraph{From treewidth to well-linked-number.}
We then consider bounding the well-linked-number in terms of treewidth.
For this we use the following lemma, which is presented explicitly in~\cite[Lemma~7.20]{DBLP:books/sp/CyganFKLMPPS15}.

\begin{lemma}[{\cite[Lemma~7.20]{DBLP:books/sp/CyganFKLMPPS15}}]
\label{lem:twbalsep}
Let $G$ be a graph and $X \subseteq V(G)$.
There exists a separation $(A,B)$ of $G$ of order $\tw(G)+1$ so that $|(A \setminus B) \cap X| \le \frac{2}{3} \cdot |X|$ and $|(B \setminus A) \cap X| \le \frac{2}{3} \cdot |X|$.
\end{lemma}

\lemwltotwlink*
\begin{proof}
Suppose that there is a well-linked set $W \subseteq E(\su{G})$ with $\bdc(W) > 3 \cdot (\tw(G)+1)$.
Let $(A,B)$ be a separation of $G$ of order $\tw(G)+1$ with $|(A \setminus B) \cap \bd(W)| \le \frac{2}{3} \cdot \bdc(W)$ and $|(B \setminus A) \cap \bd(W)| \le \frac{2}{3} \cdot \bdc(W)$, which is guaranteed to exist by \Cref{lem:twbalsep}.
Let $(C_A, C_B)$ be the bipartition of $W$ constructed by putting a hyperedge $e \in W$ to $C_A$ if $V(e) \subseteq A$ and to $C_B$ otherwise (in which case $V(e) \subseteq B$).
We have that $\bd(C_A) \subseteq (A \cap B) \cup ((A \setminus B) \cap \bd(W))$ and $\bd(C_B) \subseteq (A \cap B) \cup ((B \setminus A) \cap \bd(W))$.
It follows that
\begin{align*}
\bdc(C_A) &\le \tw(G)+1 + \frac{2}{3} \cdot \bdc(W)\\
&< \bdc(W).
\end{align*}
By a similar argument we conclude that $\bdc(C_B) < \bdc(W)$, contradicting that $W$ is well-linked.
\end{proof}

\paragraph{From well-linked-number to treewidth.}
We prove \Cref{lem:computetorsotdecomp} via two intermediate lemmas.

A \emph{branch decomposition} of a hypergraph is a superbranch decomposition where every non-leaf node has degree three.
The width of a branch decomposition $\Tc$ is $\width(\Tc) = \adhsize(\Tc)$.
We start by giving a version of \Cref{lem:computetorsotdecomp} that outputs a branch decomposition instead of a tree decomposition.
This follows the techniques of~\cite{RobertsonS-GMXIII}.

\begin{lemma}
\label{lem:wltobwoutputdecomp}
There is an algorithm that, given a hypergraph $G$ and a hyperedge $e \in E(G)$, in time $2^{\OO(\bdc(e)+\wl_e(G))} \cdot \|G\|^{\OO(1)}$ returns a branch decomposition of $G$ of width $\le \max(\bdc(e), 2 \cdot \wl_e(G))$.
\end{lemma}
\begin{proof}
We implement a recursive algorithm, that takes $G$ and $e$ as input, returns a branch decomposition $\Tc$ of $G$ of width $\le \max(\bdc(e), 2 \cdot \wl_e(G))$, and runs in time $2^{\OO(w)} \cdot \|G\|^{\OO(1)}$, where $w$ is the width of the returned branch decomposition.

The base case is when $|E(G)| \le 2$, in which case the unique branch decomposition has width $\bdc(e)$, and we can easily construct it in $\|G\|^{\OO(1)}$ time.

When $|E(G)| \ge 3$, we first apply the algorithm of \Cref{lem:testwl} to test if $E^{-} = E(G) \setminus \{e\}$ is well-linked in time $2^{\OO(\bdc(e))} \cdot \|G\|^{\OO(1)}$.

Suppose first that the algorithm concludes that $E^{-}$ is not well-linked and returns a bipartition $(C_1,C_2)$ of $E^{-}$ so that $\bdc(C_i) < \bdc(E^{-})$ for both $i \in [2]$.
Let $G_i = G \rescliqs \co{C_i}$, and denote by $e_i$ the hyperedge of $G_i$ corresponding to $\co{C_i}$.
We have that $\wl_{e_i}(G_i) \le \wl_{e}(G)$ because any well-linked set $W$ in $G_i$ not containing $e_i$ is also a well-linked set in $G$ not containing $e$.
We also have that $\bdc_{G_i}(e_i) < \bdc_G(e)$ because $\bdc(C_i) < \bdc(E^{-})$.
Therefore, we apply the algorithm recursively to compute, for both $G_1$ and $G_2$, a branch decomposition $\Tc_i = (T_i, \lmap_i)$ of $G_i$ of width at most $\max(\bdc_{G_i}(e_i), 2 \cdot \wl_{e_i}(G_i)) \le \max(\bdc_G(e), 2 \cdot \wl_e(G))$.

We construct a branch decomposition $\Tc = (T,\lmap)$ of $G$ by taking the disjoint union of $\Tc_1$ and $\Tc_2$, identifying the leaves of $\Tc_1$ and $\Tc_2$ corresponding to $e_1$ and $e_2$ into a node $t$, and adding a leaf adjacent to $t$ corresponding to $e$.
Clearly, the width of $\Tc$ is at most the maximum of the widths of $\Tc_1$ and $\Tc_2$, and $\bdc(e)$.

Suppose then that the algorithm of \Cref{lem:testwl} concluded that $E^{-}$ is well-linked.
In that case, let $C_1 \subseteq E^{-}$ be an arbitrary subset of $E^{-}$ of size $|C_1| = 1$, and $C_2 = E^{-} \setminus C_1$.
Now, define $G_1$, $G_2$, $e_1$, and $e_2$ similarly as in the previous case.
Because $C_1$ is well-linked and $|C_1| = 1$, we trivially obtain a branch decomposition $\Tc_1$ of $G_1$ of width $\wl_e(G)$.
Because $C_1$ and $E^{-}$ are well-linked, we have
\[\bdc(C_2) \le \bdc(C_1) + \bdc(e) \le \bdc(C_1) + \bdc(E^{-}) \le 2 \cdot \wl_e(G).\]
We also have that $\wl_{e_2}(G_2) \le \wl_{e}(G)$, because any well-linked set in $G_2$ not containing $e_2$ corresponds to a well-linked set in $G$ not containing $e$.
We therefore construct recursively a branch decomposition $\Tc_2 = (T_2, \lmap_2)$ of $G_2$, of width $\le \max(\bdc_{G_2}(e_2), 2 \cdot \wl_{e_2}(G_2)) \le 2 \cdot \wl_e(G)$.
By combining $\Tc_1$ and $\Tc_2$ as in the previous case, we obtain a branch decomposition of $G$ of width at most $2 \cdot \wl_e(G)$.

Clearly, both of the recursion steps can be implemented in time $2^{\OO(w)} \cdot \|G\|^{\OO(1)}$, where $w$ is the width of the resulting decomposition.
Because in both cases we have that $|E(G_i)| \le |E(G)|-1$ and $|E(G_1)|+|E(G_2)| \le |E(G)|+1$, it follows that there are at most $\OO(|E(G)|)$ recursion steps, so the overall running time is $2^{\OO(w)} \cdot \|G\|^{\OO(1)}$.
\end{proof}

We then recall the well-known fact that branch decompositions can be converted into tree decompositions.

\begin{lemma}[\cite{RobertsonS91}]
\label{lem:twbwconn}
There is an algorithm, that given a normal hypergraph $G$ and a branch decomposition $\Tc$ of $G$, in time $\|G\|^{\OO(1)}$ outputs a tree decomposition $\Tc' = (T',\bag')$ of $\primal(G)$, and an injective mapping $q \colon E(G) \to \leaves(T')$ so that
\begin{itemize}
\item $\width(\Tc') \le \frac{3}{2} \cdot \width(\Tc)-1$,
\item $\|\Tc'\| \le \|G\|^{\OO(1)}$,
\item the maximum degree of $T'$ is $3$, and
\item for all $e \in E(G)$, $V(e) \subseteq \bag(q(e))$.
\end{itemize}
\end{lemma}
\begin{proof}
Let $\Tc = (T,\lmap)$ be the given branch decomposition.
We define $\bag \colon V(T) \to 2^{V(G)}$ as $\bag(t) = V(\torso(t))$ when $t \in \vint(T)$ and $\bag(t) = V(\lmap(t))$ when $t \in \leaves(T)$.
We observe that $(T,\bag)$ is a tree decomposition of $\primal(G)$.
Furthermore, it holds that $\|(T,\bag)\| \le \OO(\|G\|^2)$, the maximum degree of $T$ is $3$, and if we set $q(e) = \lmap^{-1}(e)$, then $q \colon E(G) \to \leaves(T)$ is an injective mapping so that $V(e) \subseteq \bag(q(e))$.

Clearly, $(T,\bag)$ can be constructed from a representation of $\Tc$ in polynomial time.
It remains to bound the width of $(T,\bag)$.

Because $G$ is normal, we have that when $t \in \leaves(T)$, it holds that $|\bag(t)| = |V(\lmap(t))| = \bdc(\lmap(t)) \le \width(\Tc)$.
When $t \in \vint(T)$, we have that $|\bag(t)| = |\adh(at) \cup \adh(bt) \cup \adh(ct)|$, where $a,b,c$ are the nodes adjacent to $t$ in $T$.
If a vertex of $G$ is in one of the sets $\adh(at)$, $\adh(bt)$, and $\adh(ct)$, then it is in at least two of them, so it follows that $|\bag(t)| \le \frac{3}{2} \cdot \width(\Tc)$.
Therefore, in both cases $|\bag(t)| \le \frac{3}{2} \cdot \width(\Tc)$, so $\width((T,\bag)) \le \frac{3}{2} \cdot \width(\Tc)-1$.
\end{proof}

Now \Cref{lem:computetorsotdecomp} is a straightforward consequence.

\computetdecomp*
\begin{proof}
Follows by combining \Cref{lem:wltobwoutputdecomp} with \Cref{lem:twbwconn}.
\end{proof}

%% file: automapp.tex
\section{Tree decomposition automata}
\label{sec:apptdautom}
We then give a more formal definition of tree decomposition automata, which is based on the definition in~\cite{DBLP:conf/focs/KorhonenMNP023}.
Assume that the vertices of the graphs we process come from a countable, totally ordered universe $\Omega$, which could be assumed to equal $\mathbb{N}$.
A tree decomposition automaton is a tuple $\autom = (Q, \iota, \delta)$, where
\begin{itemize}
\item $Q$ is a (possibly infinite) set of states,
\item $\iota \colon 2^{\Omega} \times 2^{\binom{\Omega}{2}} \to Q$ is an \emph{initial mapping} that maps bags of leaf nodes to states, and
\item $\delta \colon 2^{\Omega} \times 2^{\Omega} \times 2^{\Omega} \times 2^{\binom{\Omega}{2}} \times Q \times Q \to Q$ is a \emph{transition mapping} that describes the transitions.
\end{itemize}

We assume that the state set $Q$ contains a ``null state'' $\bot$.
The run of a tree decomposition automaton $\autom$ on an annotated tree decomposition $(T,\bag,\edges)$ is the unique labeling $\run_{\autom} \colon V(T) \to Q$ satisfying,
\begin{itemize}
\item for each node $\ell$ with no children,
\[\run_{\autom}(\ell) = \iota(\bag(\ell), \edges(\ell)),\]
\item for each node $t$ with one child $x$, 
\[\run_{\autom}(t) = \delta(\bag(t), \bag(x), \emptyset, \edges(t), \run_{\autom}(x), \bot), \text{ and}\]
\item for each node $t$ with two children $x$ and $y$,
\[\run_{\autom}(t) = \delta(\bag(t), \bag(x), \bag(y), \edges(t), \run_{\autom}(x), \run_{\autom}(y)).\]
\end{itemize}

The last definition is slightly ambiguous in that it requires either an ordering of $x$ and $y$ or that the function $\delta$ is symmetric with respect to swapping the roles of $x$ and $y$.
In typical applications the symmetry can be assumed, but the argument goes through also if we fix the ordering by assuming that the nodes of $T$ are labeled with labels from $\Omega$ and choosing $x < y$.

On the algorithmic level, $\autom$ is represented as a pair of word RAM machines, one implementing the function $\iota$ and other the function $\delta$.
If $\autom$ has the property that the functions $\iota$ and $\delta$ run in time at most $\evaltime(k)$ when computing runs on tree decompositions of width at most $k$, then $\autom$ has evaluation time $\evaltime(k)$.
This also implies that the states can be represented in space of at most $\evaltime(k)$ words, i.e. $\OO(\evaltime(k) \log n)$ bits, as they are assumed to be explicitly output by these word RAM machines.